\documentclass[letterpaper, 11pt, oneside]{article}  
\newcommand{\THESIS}{}

\usepackage{verbatim}  
\usepackage{vector}  
\usepackage{array}
\usepackage{amsmath}
\usepackage{amsthm}
\usepackage{amssymb}
\usepackage{rotating}
\usepackage{hyperref}
\usepackage{subcaption}
\usepackage{multirow}
\usepackage{graphicx, balance}
\usepackage[lined,boxed]{algorithm2e}
\newcolumntype{?}{!{\vrule width 1pt}}

\ifdefined\THESIS
\newcommand{\FD}{Free Discovery}
\newcommand{\PC}{Popularity Complement}
\else
\newcommand{\FD}{FD}
\newcommand{\PC}{PC}
\fi

\hypersetup{urlcolor=blue, colorlinks=true}  

\newtheorem{theorem}{Theorem}

\title{Network Flow Based Post Processing for Sales Diversity}
\author{
        Arda Antikacioglu \\
        Department of Mathematical Sciences\\
        Carnegie Mellon University\\
        aantikac@andrew.cmu.edu\\
            \and
        R Ravi\\
        Tepper School of Business\\
        Carnegie Mellon University\\
	    ravi@cmu.edu\\
}
\date{\today}

\begin{document}
\maketitle

\begin{abstract}
Collaborative filtering is a broad and powerful framework for building recommendation systems that has seen widespread adoption. Over the past decade, the propensity of
such systems for favoring popular products and thus creating echo chambers have been observed. This has given rise to an active area of research that seeks to diversify recommendations generated by such algorithms.\cite{adomavicius2011maximizing,christoffel2015blockbusters,vargas2011rank}.

We address the problem of increasing diversity in recommendation systems that are based on collaborative filtering that use past ratings to predicting a rating quality for potential recommendations. Following our earlier work,~\cite{antikacioglu2015recommendation}, we formulate recommendation system design as a subgraph selection problem from a candidate super-graph of potential recommendations where both diversity and rating quality are explicitly optimized: (1) On the modeling side, we define a new flexible notion of diversity that allows a system designer to prescribe the number of recommendations each item should receive, and smoothly penalizes deviations from this distribution. (2) On the algorithmic side, we show that minimum-cost network flow methods yield fast algorithms in theory and practice for designing recommendation subgraphs that optimize this notion of diversity. (3) On the empirical side, we show the effectiveness of our new model and method to increase diversity while maintaining high rating quality in standard rating data sets from Netflix and MovieLens.
\end{abstract}

\section{Motivation}
Collaborative filtering has long been a favored approach in recommender systems since its recommendations are derived mainly from the record of interactions between users and items. However, a key concern of CF systems is the filter bubble, the idea that recommendation systems that focus solely on accuracy lead to echo chambers that amplify ``rich-get-richer'' effects among the recommended items \cite{celma2008hits,fleder2009blockbuster,pariser2011filter,zhou2010impact}. This problem stems from the way these systems are designed since they can only make confident recommendations on items that have had a lot of engagement, and hence increase their importance. This is the main motivation to diversify the recommendations of such CF systems.

Recent ethical concerns about algorithms have focused on similar issues about algorithmic results inherently being biased~\cite{WeaponsMathDestructionBook,npr-segment}. One approach to counter the status quo, also advocated by Karger~\cite{npr-segment} is to explicitly design algorithms that do not discriminate by designing an appropriate objective function that will increase diversity in CF recommendations.

The importance of diversifying recommendations for the sake of the user arises from their intrinsic appreciation for novelty and serendipity, a view that is supported by psychological studies \cite{tulving1994novelty}. Conversely, research in recommendation systems~\cite{mcnee2006being} has shown that focusing solely on ratings hurts user satisfaction. This has led to a subfield of recommendation systems that focuses on improving diversity for the benefit of the user \cite{vargas2011rank,zhang2008avoiding}.

A third motivation is the business need for diversifying recommendations: long-tail catalogs that are frequent in the internet \cite{brynjolfsson2011goodbye, brynjolfsson2003consumer} as well as media distributors with all-you-can-play business models \cite{goldstein2006profiting} require that the recommendations influence users to consume diverse content by driving more traffic to different portions of the site.

\noindent {\bf Roadmap:} In this paper, we address the non-diverse nature of CF recommendations, the needs of a long-tail business to shape traffic on its own site, and diversifying recommendations for the benefit of the users. We define a notion of diversity conducive to all these needs based on the degree-properties of the graph defined by the recommendations (Sections ~\ref{section:Chapter2ProblemDefinition}, ~\ref{section:Chapter2Contributions}). After reviewing related work (Section ~\ref{section:Chapter2RelatedWork}), we show that the design problem for this notion can be solved efficiently both in theory and practice using network flow techniques (Section ~\ref{section:Chapter2Algorithms}). We validate our method by showing how to adapt standard collaborative filtering algorithms with an efficient post-processing step to optimize for our measure of diversity by sacrificing very little on the recommendation quality on standard data sets (Section ~\ref{section:Chapter2Experiments}).

\section{A New Graph Optimization Problem}
\label{section:Chapter2ProblemDefinition}

We model all the user-item recommendations provided by a CF system as a bipartite graph, and the choice of recommendations actually given to the users as a subgraph selection problem in this graph. The constraints on the number of items that can be recommended to a user put bounds on the out-degree of the user nodes. Following our earlier work~\cite{antikacioglu2015recommendation}, we model the problem of achieving diversity among the items as specifying target in-degree values for each item and then finding a subgraph that satisfies these constraints as closely as possible. We develop the resulting graph optimization problem next.

We start our discussion by reviewing the well known $b$-matching problem on bipartite graphs \cite{schrijver2002combinatorial}. In the $b$-matching problem, an underlying bipartite graph $G=(L,R)$ with edge set $E$ is given, along with a nonnegative weight $g$ on the edges, and two vectors of non-negative integers $(c_1,\ldots, c_l)$ and $(a_1,\ldots, a_r)$ (degree bounds) such that $\sum_{i=1}^l c_i =\sum_{i=1}^r a_i$. The goal is to find a maximum $g$-weight (or minimum $g$-cost) subgraph $H$ where the degree of vertex $u_i\in L$ in $H$ is $c_i$ for every $1\leq i \leq l$ and the degree of vertex $v_j\in R$ in $H$ is $a_j$ for every $1\leq j\leq r$. This problem generalizes the well-known maximum weight perfect matching problem, which can be obtained as a special case if we set all target degrees to 1. Like the perfect matching problem in bipartite graphs, the $b$-matching problem can be solved by a reduction to a network flow model by adding a source with arcs to $L$, a sink with arcs from $R$ and using the degree constraints as capacities on these respective arcs.

Assume that the degree bounds are given. The vector $(c_1,\ldots, c_l)$ will be taken as a vector of hard constraints that must be met exactly, based on display constraints for the users. In other words, we consider a subgraph $H$ to be a \underline{feasible solution} if and only if  $\deg_H^+(u_i) = c_i$ for all $u_j\in L$. The vector $(a_1,\ldots, a_r)$ will be specified by the recommendation system designer to reflect the motivations described above (increase the coverage of items in CF results, increase the novelty to users on average, or shape traffic to some items). However, this target degree bounds may be unattainable (i.e. there is no feasible $b$-matching for these degree bounds). To handle this potential infeasibility, we incorporate them in the objective.
We call the vector $(a_1,\ldots, a_r)$, the \underline{target degree distribution}. We now define the objective for a given feasible solution $H$,

\[ f(H) = \sum_{v_j\in R} |deg_H^-(v_j) - a_j| \]

which is simply the sum of the violations (in both directions) of the degree constraints for $R$. We call this objective the  \underline{discrepancy} between H and the degree distribution $(a_1,\ldots, a_r)$, and we name the problem of meeting the hard constraints $(c_1,\ldots, c_l)$ while minimizing this objective as   ${\text{MIN-DISCREPANCY}(G,\{c_i\}_{i=1}^l,\{a_j\}_{j=1}^r)}.$

The min-discrepancy problem defined above generalizes the $b$-matching problem, and has objective value zero iff there is a feasible $b$-matching for the given degree bounds.
Like the weighted $b$-matching problem mentioned above, we can adapt our problem to the weighted setting where each edge has a real-valued weight. In this setting, the objective to maximize the total weight of the chosen edges, among graphs that have the minimal discrepancy possible from the given targets. 
\section{Post-processing a CF Recommender}
\label{section:Chapter2Contributions}
We now show how we can apply the graph optimization problem defined above to post-processing the results of a CF recommendation system.
As input to a CF system, we have a set of items $I$, a set of users $U$, and list of known ratings given by each user to different subsets of the items.
The CF system outputs a relevance function $rel:U\times I\to [a,b]$ that takes pairs of users and items to a predicted recommendation quality in some interval on the real line. \ifdefined\THESIS(If $g$ can be thought of as a similarity measure, then $\frac{1}{a}(b-rel(u,i))$ can be thought of as a dissimilarity measure.) \fi Without any extra information on the problem domain, CF systems employ user-based filtering, item-based filtering, matrix factorization, or other methods to arrive at these predicted rating qualities. For the rest of this paper, this rating function will be considered to be given as a black-box since its implementation details have no consequence on our model, even though we will experiment with various options in our empirical tests.

\ifdefined\THESIS
\noindent\textbf{A Surplus of Candidate Recommendations.} To generate the graph $G$ that will serve as input to our optimization problem, we choose only the recommendations for which the CF recommender predicts a rating above a certain threshold -- we enforce this by constraining each user's recommendation list to their top-$k$ candidate recommendations. This is a standard approach used in recommendation diversification \cite{adomavicius2012improving,vargas2011rank} and ensures that none of our candidate recommendations are below a certain quality level, thus establishing a quality baseline for our algorithm.

We apply our max-weight min-discrepancy method on this graph with the given predicted ratings as weights and given degree bounds as a post-processing step for the CF system results to increase their diversity. \\
\fi
\textbf{Summary of Contributions.}

\ifdefined\THESIS
\begin{enumerate}
\else
\begin{enumerate}[leftmargin=1em,labelwidth=*,align=left]
\fi
\item Following our earlier work~\cite{antikacioglu2015recommendation}, we model the problem of post-processing recommendations from a CF system to increase diversity as a maximum-weight degree-constrained subgraph selection problem to minimize the discrepancy from a target distribution.

\item We demonstrate that the problem of finding maximum-weight min-discrepancy subgraph can be reduced to the problem of finding minimum cost flows. In particular, this shows that the discrepancy between a recommendation system and {\em any} desired indegree distribution can be minimized in polynomial time. The abundance of fast solvers \cite{kovacs2015minimum} for this problem makes our model not just theoretically interesting, but also practically feasible. Moreover, we prove that aggregate diversity maximization can be implemented special case of the discrepancy minimization problem. This generalizes the work of Adomavicius and Kwon on maximizing aggregate diversity \cite{adomavicius2011maximizing} while simultaneously maximizing recommendation quality, while matching the same asymptotic runtimes.

\item We conduct experiments on standard datasets such as MovieLens-1m, and Netflix Prize data. By feeding our discrepancy minimizer as a post-processing step on the undiversified recommendation networks created by standard collaborative filtering algorithms, we measure the trade-off our algorithm makes between discrepancy and recommendation quality under a variety of parameter settings. We compare against baselines and other diversification approaches, and find that our diversifier makes more relevant recommendations despite achieving higher diversity gains, as measured not only by our discrepancy measure, but also by standard sales diversity metrics such as the Gini index or aggregate diversity.
\end{enumerate} 

\section{Related Work}
\label{section:Chapter2RelatedWork}

First we review related work on various collaborative filtering approaches, and then discuss various extant notions of diversity already considered in the recommender system literature.
\subsection{Collaborative Filtering}
Collaborative filtering is the most versatile and widely accepted way of building recommender systems. The main idea behind collaborative filtering is to exploit the similarities between different users or between different items using user feedback. While there are many different methods for doing this, we constrain our evaluation to three representative approaches.

\ifdefined\THESIS
\begin{enumerate}
\item Matrix factorization approaches assume the existence of $D$ latent features which describe both users and items, and seeks to find two rank $D$ matrices whose product approximates the matrix of all known rankings. The advantage is that $D$ is typically much smaller than the number of users or items. In our work, we experiment with a version of this approach due to Hu \cite{hu2008collaborative}.

\item Another popular approach is neighborhood based recommenders, which can either be user-based or item-based. These approaches define a distance between pairs of users and pairs of items respectively, using measures like cosine similarity or Pearson correlation \cite{desrosiers2011comprehensive}. The user-based approach then predicts the unknown rating from user $u$ to item $i$ by taking a distance weighted linear combination of the ratings of similar users on item $i$. The item-based approach operates similarly, but instead takes a weighted combination of the ratings of user $u$ on items similar to item $i$. We use the implementations of these methods in RankSys \cite{sandoval2015novelty} in our experiments.

\item Finally, we consider a graph based recommender strategy due to Cooper et. al. \cite{cooper2014random}. This method considers a bipartite graph of known user and item interactions, ignoring all rating information. In this graph, a random walk of length 2 from a user $u$ corresponds to the selection of a user similar to $u$, in the sense that both $u$ and any user reachable from $u$ in 2 steps have at least one item as a common interest. Therefore, a random walk of length 3 corresponds to sampling an item liked by a similar user, and recommendations for a user $u$ are ranked according to how many random walks of length 3 starting at that user terminate at a given item. Since this method is both simple to state and implement on small to medium sized datasets, we use our own implementation of this method in our empirical comparisons. While less commonly used than the first two types of recommenders we discussed, this approach is still representative of a large class of recommendation strategies such as UserRank \cite{gao2011userrank}, ItemRank \cite{gori2007itemrank}, or other other random walk based techniques \cite{liu2012solving}.    
\end{enumerate}
\else
Matrix factorization approaches assume the existence of $D$ latent features which describe both users and items, and seeks to find two rank $D$ matrices whose product approximates the matrix of all known rankings. The advantage is that $D$ is typically much smaller than the number of users or items. In our work, we experiment with a version of this approach due to Hu \cite{hu2008collaborative}.

Another popular approach is neighborhood based recommenders, which can either be user-based or item-based. These approaches define a distance between pairs of users and pairs of items respectively, using measures like cosine similarity or Pearson correlation \cite{desrosiers2011comprehensive}. The user-based approach then predicts the unknown rating from user $u$ to item $i$ by taking a distance weighted linear combination of the ratings of similar users on item $i$. The item-based approach operates similarly, but instead takes a weighted combination of the ratings of user $u$ on items similar to item $i$. We use the implementations of these methods in RankSys \cite{sandoval2015novelty} in our experiments.

Finally, we consider a graph based recommender strategy due to Cooper et. al. \cite{cooper2014random}. This method considers a bipartite graph of known user and item interactions, ignoring all rating information. In this graph, a random walk of length 2 from a user $u$ corresponds to the selection of a user similar to $u$, and a random walk of length 3 corresponds to sampling an item liked by a similar user. Recommendations for a user $u$ are ranked according to how many random walks of length 3 starting at that user terminate at a given item. Since this method is both simple to state and implement on small to medium sized datasets, we use our own implementation of this method in our empirical comparisons. While less commonly used than the first two types of recommenders we discussed, this approach is still representative of a large class of recommendation strategies such as UserRank \cite{gao2011userrank}, ItemRank \cite{gori2007itemrank}.
\fi

\subsection{Sales Diversity}
\ifdefined\THESIS
\textbf{User-Focused Diversity:} User novelty has been called intra-list diversity \cite{castells2011novelty}, with the list referring to the list of recommendations made to a particular user. The need for novelty from the user's point of view is a psychological one \cite{tulving1994novelty}. While lack of intra-list diversity was a particularly bad shortcoming of early recommender systems \cite{ali2004tivo}, these problems have since been addressed in many works \cite{vargas2015novelty}. In this work, we do not consider diversity at a user level, and instead take a system level view, which motivated more by business needs than user needs.
\fi

\textbf{The Need for Sales Diversity:} As mentioned above, the need for system-level diversification in recommenders is a business related one. Since the internet enables businesses with low inventory costs, focusing on making more recommendations in the long tail can be an effective retail strategy. This view is most clearly expressed by Anderson, who advocates selling ``selling less of more'' \cite{Anderson2006}. Interestingly, recommender systems rarely capitalize on this opportunity, and often compound the problems observed with popularity bias. Indeed, Zhou et. al. find that YouTube's recommendation module leads to an increase in popularity for the most popular items \cite{zhou2010impact}. Similarly, Celma et. al. report similar findings for music recommendations on Last.fm~\cite{celma2008hits}. Hosanagar and Fleder show that this popularity bias can lead to subpar pairings between users and items, potentially hurting customer satisfaction \cite{fleder2009blockbuster}, and McNee reports that a focus on accuracy alone has hurt the user experience of recommender systems \cite{mcnee2006being}. Since recommendations have an outsized impact on customer behavior \cite{senecal2004influence, adomavicius2011maximizing}, businesses have a need to control the distribution of recommendations that they surface in their recommenders.

\ifdefined\THESIS
\textbf{Metrics for Sales Diversity:}
There are several well-established metrics for measuring sales diversity, and we focus our attention on three. 

\begin{enumerate}
\item The most popular among these is the aggregate diversity, which is the total number of objects that have been recommended to at least one user. Under this name, this measure has been used notably by Adomavicius and Kwon \cite{adomavicius2012improving, adomavicius2011maximizing} and by Castells and Vargas \cite{vargas2015novelty}. It has also been used as a measure of system-wide diversity under the name of coverage \cite{ge2010beyond, adamopoulos2011unexpectedness}. While easy to understand and measure, the aggregate diversity leaves a lot be desired as a measure of distributional equality. In particular, aggregate diversity treats an item which was recommended once as well-covered as an item which was recommended thousands of times. For example, imagine a system that recommends each item in a set of $n$ items twice. This network will have the same aggregate diversity as a network which recommends one of the items $n$ times, and every other item only once, even though this system is much more biased than the first. Moreover, aggregate diversity can be a misleading measure of diversity when the number of users far outnumbers the size of the catalog. Under these circumstances, even very obscure items may get recommended at least once. Therefore, while aggregate diversity is a good baseline, more refined measures are needed to evaluate the equitability of the distribution of recommendations.

\item An example of a more nuanced metric is provided by the Gini index. This measure is most popularly used in economics, as a quantization of wealth or income inequality. The Gini Index can be adapted for the recommendation setting by considering the number of recommendations an item gets as its ``wealth'' in the system. The Gini index defines the most equitable distribution to be the one where every item is recommended an equal number of times. Given the actually realized distribution of recommendations, it aggregates the difference between the number of recommendations the bottom $\text{n}^{th}$ percentile gets in the system and the number of recommendations they would have obtained under the uniform distribution where $n$ ranges from 0 to 100\%. The measure we propose is a particularly good proxy for the Gini index, since both measure a notion of distance from the uniform distribution. Since recommender systems produce distributions even more unequal than the typical wealth distributions within a country, this metric has found widespread acceptance in the recommendation community \cite{shani2011evaluating, ge2010beyond, herlocker2004evaluating, ren2014avoiding}.

\item Finally, we consider the entropy of the distribution of recommendations. Entropy has its roots in physics and information theory, where it is used to measure the amount of information contained in a stochastic process. For every item, we can define a probability of being surfaced by the recommender by counting what fraction of recommendations (made to any user) point to this item. As with the  Gini index, optimal entropy is achieved if and only if the recommendation distribution is uniform. While less common than either aggregate diversity or the Gini index, the entropy of the recommender system has also been used by many researchers \cite{shani2011evaluating, szlavik2011diversity}.
\end{enumerate}
\else
\textbf{Metrics for Sales Diversity:} There are several well-established metrics for measuring sales diversity, and we focus our attention on three. The most popular among these is the aggregate diversity, which is the total number of objects that have been recommended to at least one user. Under this name, this measure has been used notably by Adomavicius and Kwon \cite{adomavicius2012improving, adomavicius2011maximizing} and by Castells and Vargas \cite{vargas2015novelty}. It has also been used as a measure of system-wide diversity under the name of coverage \cite{ge2010beyond, adamopoulos2011unexpectedness}. While easy to understand and measure, the aggregate diversity leaves a lot be desired as a measure of distributional equality. In particular, aggregate diversity treats an item which was recommended once as well-covered as an item which was recommended thousands of times. For example, imagine a system that recommends each item in a set of $n$ items twice. This network will have the same aggregate diversity as a network which recommends one of the items $n$ times, and every other item only once, even though this system is much more biased than the first.

An example of a more nuanced metric is provided by the Gini index. This measure is most popularly used in economics, as a quantization of wealth or income inequality. The Gini Index can be adapted for the recommendation setting by considering the number of recommendations an item gets as its ``wealth'' in the system. The Gini index defines the most equitable distribution to be the one where every item is recommended an equal number of times. Given the actually realized distribution of recommendations, it aggregates the difference between the number of recommendations the bottom $\text{n}^{th}$ percentile gets in the system and the number of recommendations they would have obtained under the uniform distribution where $n$ ranges from 0 to 100\%. The measure we propose is a particularly good proxy for the Gini index, since both measure a notion of distance from the uniform distribution. Since recommender systems produce distributions even more unequal than the typical wealth distributions within a country, this metric has found widespread acceptance in the recommendation community \cite{shani2011evaluating, ge2010beyond, herlocker2004evaluating, ren2014avoiding}.

Finally, we consider the entropy of the distribution of recommendations. Entropy has its roots in physics and information theory, where it is used to measure the amount of information contained in a stochastic process. For every item, we can define a probability of being surfaced by the recommender by counting what fraction of recommendations (made to any user) point to this item. As with the  Gini index, optimal entropy is achieved if and only if the recommendation distribution is uniform. While less common than either aggregate diversity or the Gini index, the entropy of the recommender system has also been used by many researchers \cite{shani2011evaluating, szlavik2011diversity}.
\fi

We measure the diversification performance of our methods and the baselines we test in our experimental section by all three of these metrics - aggregate diversity, Gini index and entropy.

\ifdefined\THESIS
\textbf{Approaches for Increasing Sales Diversity:} Attempts at increasing sales diversity fall into two approaches: optimization and reranking. 

\begin{enumerate}

\item The optimization approach has been taken up most notably by Adomavicius and Kwon \cite{adomavicius2011maximizing}, who consider heuristic and exact algorithms for improving aggregate diversity. Their flow based solution is approximate, while their exact solution to this problem relies on integer programming and has exponential complexity. Our work in this paper subsumes these approaches by giving an exact polynomial algorithm for aggregate diversity maximization. To the best of our knowledge, neither the Gini index nor the entropy of the degree distribution can be optimized in an exact sense.

\item The reranking based approaches are by far the more popular choice in increasing sales diversity. Here, we consider three different approaches by Castells and Vargas, spread across two different papers. The first two approaches model discovery in a recommender system by associating each user-item pair with a binary random variable called $seen$ which represents the event that a user is familiar with the given item, with the assumption that an unseen item is novel to the user. Given a training dataset $\theta$ which maps pairs of users from $U$ and items from $I$ to known rating values, the authors use a maximum likelihood model to estimate the probability that an item is seen by a user. In particular, they set the probability of an item $i$ already being known to a user as the fraction of users who have rated that item

\[ p(seen | i,\theta) = \frac{|\{u \in U | r(u,i) \ne\emptyset \}|}{|U|}\] 

Given these probabilities, the popularity complement similarly defines the novelty of an item $i$ as $\text{nov}_{PC}(i|\theta) = 1- p(seen | i,\theta)$. The free discovery method of measuring novelty similarly defines the novelty of an item $i$ to a user $u$ as $\text{nov}_{FD}(i|\theta) = -\log_2(p(seen | i,\theta))$. After these probabilities are estimated from the training data, the candidate list of recommendations are reranked according to the a score which is the average of the predicted relevance of the item and the novelty of the item \cite{vargas2011rank}.
 
Instead of combining novelty and relevance components in the same function, the Bayes Rule method by the same authors explicitly adjusts the rating prediction function. In particular, let $rel:U\times I \to [0,1]$ be the function which predicts the strength of a recommendation between user and item pairs. The authors suggest that the probability that an item is relevant to a user is proportional to its predicted rating, and verify this claim experimentally. Combining this assumption with Bayesian inversion, they come up with a revised prediction function $rel_{BR}(u,i) = rel(u,i) \left(\sum_{u'} rel(u',i) \right)^{-\alpha}$, and rerank the recommendations according to this function. Note that predicted quality of a recommendation only differs from the original prediction by a factor of $\left(\sum_{u'} rel(u',i) \right)^{-\alpha}$. This term is a function of the item $i$, and grows larger as the sum of the predicted ratings of $i$ goes lower. Therefore, rescoring in this manner dampens the predicted ratings of popular items, while increasing the predicted ratings of less popular items \cite{vargas2014improving}.    
\end{enumerate}
\else

\textbf{Approaches for Increasing Sales Diversity:} Attempts at increasing sales diversity fall into two approaches: optimization and reranking. The optimization approach has been taken up most notably by Adomavicius and Kwon \cite{adomavicius2011maximizing}, who consider heuristic and exact algorithms for improving aggregate diversity. Their flow based solution is approximate, while their exact solution to this problem relies on integer programming and has exponential complexity. Our work in this paper subsumes these approaches by giving an exact polynomial algorithm for aggregate diversity maximization. To the best of our knowledge, neither the Gini index nor the entropy of the degree distribution can be optimized in an exact sense.

The reranking based approaches are by far the more popular choice in increasing sales diversity. Here, we consider three different approaches by Castells and Vargas, spread across two different papers. The first two approaches model discovery in a recommender system by associating each user-item pair with a binary random variable called $seen$ which represents the event that a user is familiar with the given item, with the assumption that an unseen item is novel to the user. Given a training dataset $\theta$ which maps pairs of users from $U$ and items from $I$ to known rating values, the authors use a maximum likelihood model to estimate the probability that an item is seen by a user. In particular, they set the probability of an item $i$ already being known to a user as the fraction of users who have rated that item. Given these probabilities, the popularity complement similarly defines the novelty of an item $i$ as $\text{nov}_{PC}(i|\theta) = 1- p(seen | i,\theta)$. The free discovery method of measuring novelty similarly defines the novelty of an item $i$ to a user $u$ as $\text{nov}_{FD}(i|\theta) = -\log_2(p(seen | i,\theta))$. After these probabilities are estimated from the training data, the candidate list of recommendations are reranked according to the a score which is the average of the predicted relevance of the item and the novelty of the item \cite{vargas2011rank}.

Instead of combining novelty and relevance components in the same function, the Bayes Rule method by the same authors explicitly adjusts the rating prediction function. In particular, let $rel:U\times I$ be the function which predicts the strength of a recommendation between user and item pairs. The authors suggest that the probability that an item is relevant to a user is proportional to its predicted rating, and verify this claim experimentally. Combining this assumption with Bayesian inversion, they come up with a revised prediction function $rel_{BR}(u,i) = rel(u,i) \left(\sum_{u'} rel(u',i) \right)^\alpha$, and rerank the recommendations according to this function. Note that predicted quality of a recommendation only differs from the original prediction by a factor of $\left(\sum_{u'} rel(u',i) \right)^{-\alpha}$. This term is a function of the item $i$, and grows larger as the sum of the predicted ratings of $i$ go down. \cite{vargas2014improving}.
\fi

\ifdefined\THESIS
\textbf{Constrained Recommendations:} Finally, our problem is one of constrained recommendation, and this problem has been studied in contexts where the data is both large and small. In the small data end of the of the spectrum, recommenders have been developed to solve problems such as matching students to courses based on prerequisites and requirements  \cite{parameswaran2011recommendation} or matching reviewers to paper submissions \cite{mimno2007expertise, karimzadehgan2009constrained}. These types of models perform well because they are built specifically for the task at hand, but they are not suited for general purpose recommendation tasks for increasing diversity.

On the other end of the spectrum, in a context of matching buyers with sellers with the goal of maximizing revenue~\cite{chen2016conflict}, the problem has been modeled as a matching problem to maximize the number of recommended edges (rather than their diversity). Similarly, large-scale matching problems modeling recommendations have been tackled in a distributed setting \cite{makari2013distributed} using degree bounds, but these methods are unable to accurately enforce degree lower bounds on the items begin recommended (so they are simply set to zero). Here again, the objective is to maximize the number of edges chosen rather than any measure of diversity. While our methods do not scale to the same level, we are able to model diversity with degree bounds more accurately.
\fi

\section{Algorithms}
\label{section:Chapter2Algorithms}
In this section, we prove that discrepancy from a target distribution can be minimized efficiently by reducing this problem to one invocation of a minimum cost flow problem. This result holds regardless of the target in-degree distribution and the required out-degree distribution.

\subsection{Construction of the Flow Network}

Let $G=(L,R,E)$ be the input bipartite graph which contains candidate recommendations. We construct a flow network out of $G$ such that the min-cost feasible flow will have cost equal to the min-discrepancy. Our network will have $|V|+2$ nodes: two special sink nodes $t_1$ and $t_2$, as well as a copy of each node in $G$ (See Figure~\ref{fig:flownetwork}). We set the supply of each node $u_i$ to $c_i$ (its specified out-degree), and the demand of the sink $t_2$ to $\sum_{j=1}^r a_j = \sum_{i=1}^l c_i$. Next, for each arc $(u_i, v_j)\in G$, we create an arc $(u_i,v_j)$ in the flow network, with unit capacity and zero cost. For each node $v_j$, we create an arc to each sink. We add the arc $(v_j,t_1)$ of capacity $a_j$ (its target out-degree) and zero cost, and the arc $(v_j,t_2)$ of infinite capacity and cost $2$. We finally add to our network an arc $(t_1,t_2)$ between the two sinks, with infinite capacity and zero cost. Our assumptions ensure that total supply, $\sum_{i=1}^l c_i$  meets total demand $\sum_{j=1}^r a_j$, and that a feasible flow exists since each node $u_i$ in $L$ can send as much as $deg^+_G(u_i) \geq c_i$ flow to the sink $t_2$ via any $c_i$ different neighbors. Note that there are $|E|+2|R|+1$ arcs in total in our flow network.  The complete flow network constructed this way is shown in Figure~\ref{fig:flownetwork}.

\ifdefined\THESIS
\begin{figure}
\includegraphics[width=.99\textwidth]{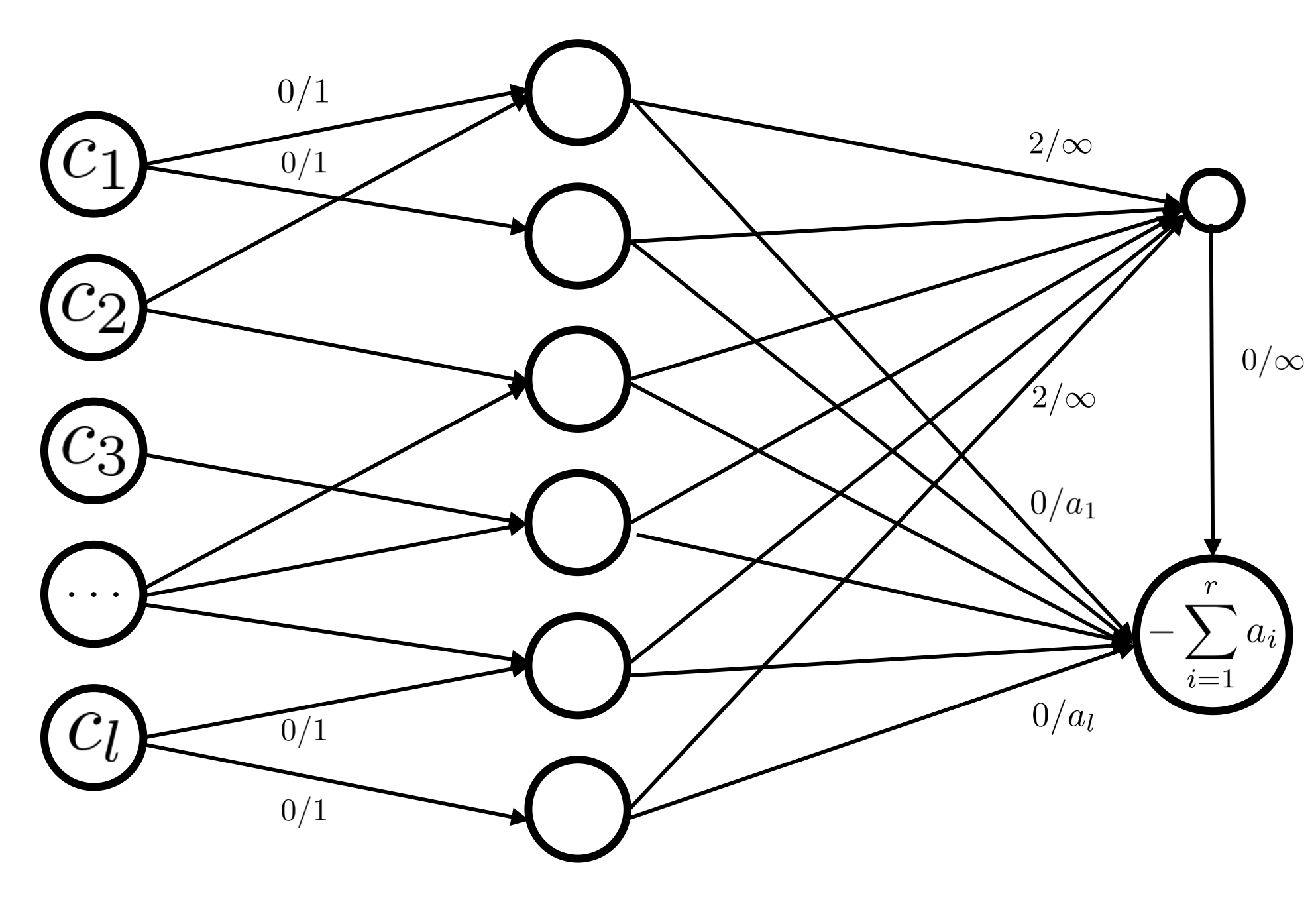}
\caption{The network flow model for min-discrepancy with nodes labelled with their supply and arcs labeled with their cost/capacity. Unlabelled nodes have zero supply.}
\label{fig:flownetwork}
\end{figure}
\else
\begin{figure}
\includegraphics[width=.99\columnwidth]{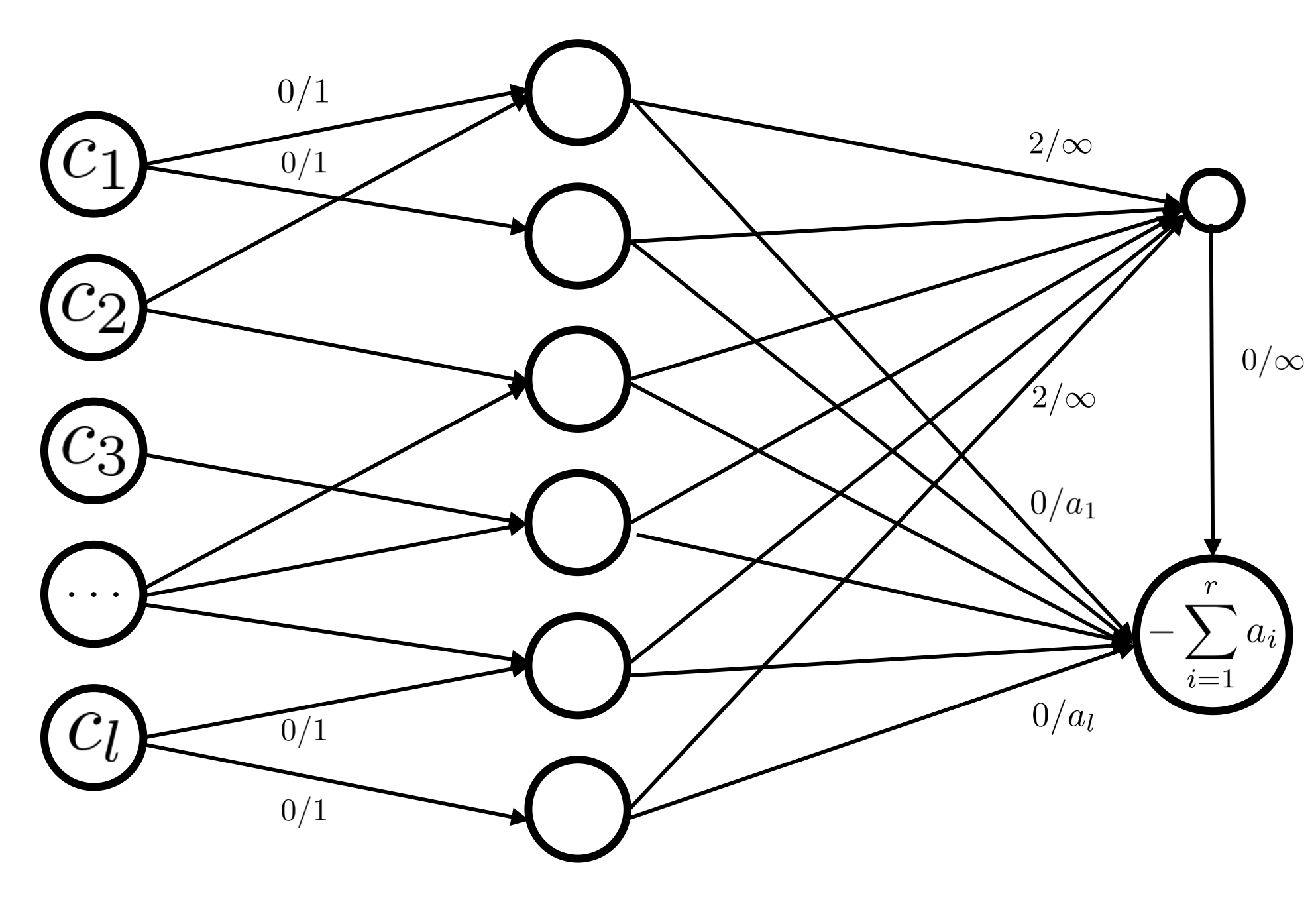}
\caption{The network flow model for min-discrepancy with nodes labelled with their supply and arcs labeled with their cost/capacity. Unlabelled nodes have zero supply.}
\label{fig:flownetwork}
\vspace{-.7cm}
\end{figure}
\fi

Our main theorem shows that the minimum cost of a flow in this network is the same as the minimum discrepancy a subgraph of $G$ has from our target in-degree distribution.

\begin{theorem}
\label{thm:main}
Suppose $G=(V,E)$ satisfies $deg^+_G(u_i) \geq c_i$ for all $u_i \in L$ and the degree distributions satisfy $\sum_{i=1}^l c_i =\sum_{i=1}^r a_i$. Then the minimum cost flow in the network constructed above has value $\text{MIN-DISCREPANCY}(G(L,R),\{c_i\}_{i=1}^l,\{a_j\}_{j=1}^r)$ and can be computed $O(|E||V|^2\log(|V|))$ time.
\end{theorem}
\begin{proof}
Consider a minimum cost flow in this network. Since the network's capacities and supplies are all integral, we may assume that this minimum cost flow is integral as well~\cite{ahuja1993network}. This means each edge crossing from $L$ to $R$ is either fully used or unused because it is unit capacity.

We let $H$ be a subgraph of $G$ defined by taking the edges of the form $(u_i,v_j)$ where $(u_i,v_j)$ is used in the flow. Since each such edge is either used or unused, and the supply of node $u_i$ is $c_i$, we will satisfy the constraints of the form $\deg_H^+(u_i) = c_i$. To see that the cost of this flow is the same as the cost of our objective, note that we can partition the vertices in $R$ into two halves: $P$ for the vertices satisfying their degree requirement $\deg_H^-(v_j) \geq a_j$, and $N$ for the vertices not satisfying their degree requirement. We can now write our objective as follows.

\begin{align*}
\sum_{v_j\in R} |\deg_H^-(v_j) - a_j| = & \sum_{v_j\in P} \left( \deg_H^-(v_j) - a_j\right) + \\
                                                                    & \sum_{v_j\in N} \left( a_j - \deg_H^-(v_j)\right)
\end{align*}

However, note that our flow is feasible. Therefore, the total number of edges recommended is $\sum_{i=1}^l c_i$. It now follows that $\sum_{v_j\in R} (\deg_H^-(v_j) - a_j) = \sum_{v_j\in R} \deg_H^-(v_j) - \sum_{i=1}^l c_i = 0$ from our assumption that $\sum_{i=1}^l c_i =\sum_{j=1}^r a_j$. Adding this to the expression above gives the following.

\[ \sum_{v_j\in R} |\deg_H^-(v_j) - a_j| = 2\sum_{v_j\in P} \left( \deg_H^-(v_j) - a_j \right) \]

In our formulation, we only pay for the flow going through a node $v_j$ if the flow is in excess of $a_j$. Since we pay $2$ units of cost for each unit of this type of flow, and don't pay for anything else, our objective matches that of the flow problem.

By reversing our reduction, we can show that every subgraph $H$ with the desired properties induces a flow with the same cost as well. Therefore, the minimum discrepancy problem can be solved by a single invocation of a minimum cost flow algorithm, on a network with $|L|+|R|+2$ nodes and $2|R|+|E|+1$ edges with capacity bounded by $|V|$. This can be solved in $O(|E||V|^2\log(|V|))$ time, using capacity scaling or other competitive methods~\cite{orlin1997polynomial}.
\end{proof}

\ifdefined\thesis
\subsection{Aggregate Diversity}
\else
\textbf{Aggregate Diversity.}
Recall that aggregate diversity is the total number of items recommended by a recommender system. Aggregate diversity does not correspond to discrepancy from any target distribution, however it can be maximized by our model as well. 

\begin{theorem}
\label{thm:aggdiv}
Suppose $\sum_{i=1}^l c_i \geq r = |R|$. Then the minimum cost flow solution in the network constructed for the $\text{MIN-DISCREPANCY}(G,\{c_i\}_{i=1}^l,\{1\}_{j=1}^r)$ problem  maximizes aggregate diversity.
\end{theorem}
\ifdefined\THESIS
\begin{proof}
The sufficiency of our condition is obvious as it is needed to make sure that the supply of the nodes in $L$ can be absorbed by the sink node. Now suppose that a recommender system achieves aggregate diversity $r$. A total of $\sum_{i=1}^l c_i$ units of flow make it to the sink, and each has to travel through an arc of cost 0 or 1. Since there are $r$ different items in $R$, and each can send 1 unit of flow without cost, this solution has cost $\left( \sum_{i=1}^l c_i \right) - r$. Conversely, suppose some solution obtains cost $\left( \sum_{i=1}^l c_i \right) - r$. The only way the cost can be reduced below $\left( \sum_{i=1}^l c_i \right)$ is achievable is through the use of 0 cost arcs. Since each such arc has capacity 1, at least $r$ such arcs must be used in the solution. This implies that a solution of aggregate diversity $r$ exists.
\end{proof}
\else
We omit the proof due to space constraints, as this is an easy corollary of Theorem \ref{thm:main}.
\fi
\subsection{Incorporating Recommendation Relevance}
\subsubsection{Cumulative Gain}
Note that we have had to assign zero costs to all the edges crossing between the two sides of our bipartition in order for our reduction to work. Recommendation strengths can be taken into account by our flow based methods, and we can find the graph that has the highest total recommendation quality given a discrepancy value using an extra pass with a flow solver. This can be done accomplished as follows: first, we solve the regular discrepancy problem, and finding the lowest discrepancy value $OPT$ attainable by the underlying $G$. Knowing this value, we can now fix the flow between $t_1$ and $t_2$ in the original flow network to $lc-OPT$, where $lc$ is the total out-degree of the subgraph from $L$. This constrains the flow solver to choose subgraphs where exactly $OPT$ of the recommendations go over the charged edges. We then keep all the other capacities the same and add new nonzero weights reflecting recommendation quality while removing all other costs. In a second pass, we find the highest weight flow in this network, which corresponds to the recommender graph with $OPT$ discrepancy with the highest total recommendation quality. We call this approach the \underline{two-pass method}\footnote{This follows the goal-programming methodology for two-objective functions, popular in Operations Research.}. Maximizing average recommendation quality in this fashion corresponds to the finding the recommendation subgraph with the highest cumulative gain. 

\ifdefined\THESIS
In some cases, recommendation algorithms do not give each recommendation a score that falls in a uniform interval, and instead rank the resulting recommendations among themselves. In this case, we cannot use the average recommendation quality as a measure of the quality of our recommendations. An appropriate measure of quality in this case is the precision-in-top-$c$ metric, where the quality of the recommendations made to the user $u$ is measured as $|N(u) \cap T_k(u)| /c$, where $T_k(u)$ is the list of top-$k$ recommendations for the user $u$. When $c$ recommendations are made for each user, the average recommendation quality of the system is

\[ \sum_{u\in L} \frac{|N(u) \cap T_k(u)|}{c} = \frac{1}{c}\sum_{u\in L} |N(u) \cap T_k(u)|\]

Note that the quantity inside the sum is simply a linear function of the recommendations made by the subgraph: recommendation edges which are in the top-$k$ for a user $u$ have weight 1, while every other edge has weight 0. Therefore, we can optimize the average of the accuracy-in-top-$k$ using a flow model as well. The objective we maximize in this setting corresponds to the average cumulative gain in the binary relevance setting.

\subsubsection{Discounted Cumulative Gain}
Not every recommendation in a list is considered equally valuable, and the value of recommendation slot depends on its rank among the presented recommendations. Discounted cumulative gain accomplishes this by weighing the relevance of the $\text{i}^{th}$ slot by $1/\log(i+1)$. That is, if we let $v_1,\ldots, v_c$ be the recommendations made to user $u$ and $rel(u,v_i)$ the relevance of the $\text{i}^{th}$ recommendation

\[ CDG(u) = \sum_{i=1}^c \frac{rel(u,v_i)}{\log(i+1)} \]

We first show how to maximize $CDG$ in the binary relevance case, which has a simpler construction than the general case.

\begin{theorem}
\label{thm:binaryCDG}
The recommendation graph matching the minimum discrepancy given by $\text{MIN-DISCREPANCY}(G(L,R),\{c_i\}_{i=1}^l,\{a_j\}_{j=1}^r)$ while having the highest cumulative discounted gain in the binary relevance setting and can be computed with one extra min-cost flow problem.
\end{theorem}
\begin{proof}
We use the construction in Theorem~\ref{thm:main} as our starting point and set the cost of each arc to 0. We fix the flow between $t_1$ and $t_2$ to $lc-OPT$ to constrain the solver to solutions which have the desired discrepancy as discussed above.

Note that when $k$ relevant recommendations are made to a user, then the resulting discounted cumulative gain is $\sum_{i=1}^k \frac{1}{\log(i)} $. In order to be able to charge this quantity in our flow model, we create create an intermediary node $n_{u,c}$. Every recommendation $(u,v)$ which has binary relevance 1, now connects $n_{u,c}$ to $v$ instead of $u$ and $c$. We also connect $n_{u,c}$ to node $u$ by using $c$ parallel arcs with costs $-1,-1/\log(2),\ldots,-1/\log(c)$, each of capacity 1. This modification only adds an extra node and $c$ arcs for each user. Furthermore, the cost of the flow is the negation of the CDG function summed across all users, and this can be minimized with a single invocation of a min-cost flow solver.
\end{proof}

The rationale behind discounting the value of later recommendations is based on a a model of the user's consumption of the recommendations: later recommendations are less likely to receive attention from the user, which diminishes their usefulness. However, the discounting serves another beneficial purpose in our model. In particular, the value of making $k$ relevant recommendations in the binary relevance model is

\[ \sum_{i=1}^k \frac{1}{\log(k)} = \Theta\left(\frac{k}{\log(k)}\right)\]

This function grows slower than linearly, which means that there are diminishing returns as more and more relevant recommendations are made to the same user. It is therefore, more advantageous to make a second relevant recommendation for a user $u$ than to make a tenth relevant recommendation to a user $u'$. This can be used to ensure that no user gets a disproportionate share of irrelevant recommendations.

\begin{theorem}
\label{thm:fullCDG}
The recommendation graph matching the minimum discrepancy given by $\text{MIN-DISCREPANCY}(G(L,R),\{c_i\}_{i=1}^l,\{a_j\}_{j=1}^r)$ while having the highest cumulative discounted gain and can be computed with one extra min-cost flow problem.
\end{theorem}
\begin{proof}
We use the construction in Theorem~\ref{thm:main} as our starting point and set the cost of each arc to 0. We fix the flow between $t_1$ and $t_2$ to $lc-OPT$ to constrain the solver to solutions which have the desired discrepancy as discussed above.

Since cumulative discounted gain depends on the ranking of the recommendations made, we create $c$ nodes for each user $n_{u,1},\ldots, n_{u,c}$. We connect each of these to the user node $u$ by an arc of cost 0 and capacity 1. For each candidate recommendation $(u,v)\in G$ we create $c$ arcs in total. For each $1\leq i \leq c$, the node $n_{u,i}$ is connected to node $v$ by a cost $-rel(u,v)/\log(i+1)$ capacity 1 arc. The rest of the construction remains the same, and these modifications add $|V|$ vertices to the construction, as well as $(c-1)|E|$ additional arcs since every recommendation edge now has $c$ copies instead of just 1.

In a feasible solution, no more than $c$ recommendations can be made to a user because the total capacity of the arcs coming out of node $u$ is $c$. The single unit of flow that is allowed to leave $n_{u,i}$ discounts the value of this recommendation by $\log(i+1)$ due to the way we set the costs. The sum of these contributions gives us the negation of the cumulative discounted gain for the system, which can be minimized with one extra min-cost flow invocation.
\end{proof}

Table ~\ref{table:algorithmSummary} summarizes the number of arcs and nodes in each of our constructions.

\begin{table}[h]
\resizebox{\textwidth}{!}{
\begin{tabular}{l|l|l|}
\cline{2-3}
                                                          & Arcs                  & Nodes        \\ \hline
\multicolumn{1}{|l|}{Cumulative Gain (Binary)}            & $2|R|+|E|+1$          & $|L|+|R|+2$  \\ \hline
\multicolumn{1}{|l|}{Cumulative Gain}                     & $2|R|+|E|+1$          & $|L|+|R|+2$  \\ \hline
\multicolumn{1}{|l|}{Cumulative Discounted Gain (Binary)} & $(c+1)|L|+2|R|+|E|+1$ & $2|L|+|R|+2$  \\ \hline
\multicolumn{1}{|l|}{Cumulative Discounted Gain}          & $|L|+2|R|+c|E|+1$     & $c|L|+|R|+2$ \\ \hline
\end{tabular}
}
\caption{The number of nodes and arcs in each of our difference relevance models. The non-discounted models are the most efficient, followed by the binary cumulative discounted model. The full discounted gain model is likely to be prohibitively expensive for most settings of $c$.}
\label{table:algorithmSummary}
\end{table}

\fi

\subsubsection{Bicriteria Optimization}
In each of the constructions above, we needed to make an extra pass with a flow solver in order to find a solution with a high level of relevance. If we used these cost settings along with the cost settings we used in \ref{thm:main}, we would no longer be optimizing only for ratings, or only for discrepancy. Instead, this results in a bicriteria objective of the form $\text{discrepancy}(H) -\mu\cdot \text{rel}(H)$, where $\mu$ can be any real number, and where relevance of a solution graph denotes the average relevance of the recommendations in $H$ as measured by any of the metrics discussed above. We call this approach \underline{the weighted method}, and demonstrate that while it is strictly worse than the two-pass method in theory, it yields acceptable results in practice while saving an extra pass of flow minimization. We discuss the performance differences in our experimental section. 
\ifdefined\THESIS

\subsection{Category Level Constraints}
It is sometimes desirable to set multiple goals in an optimization problem. For example, a news website might have a target distribution for the articles in mind, but might also want to ensure that none of the different categories such as current events, politics, sports, entertainment, etc. are neglected. Alternatively, due to a payment from a sponsor, a vendor may wish to boost the profile of a specific subset of movies in the catalog, and might need to balance its own needs about the distribution of recommendations with their commitment to their sponsor.

This type of problem can be accommodated by our model in the following way. Let there be $k$ categories $C_1,\ldots, C_k$ that partition the items in $R$ with minimum targets $A_1, \ldots, A_k$ respectively. We will require have $A_t \leq \sum_{v_j \in A_t} a_j$, i.e., the category requirements are less stringent than the aggregate of the individual target requirements. For ease of notation, let $D_1,\ldots, D_k$ be the number of times an item from category $C_1,\ldots, C_k$ are recommended. In this setting, we can optimize the objective  $\sum_{v_i\in R} |\deg_H^-(v_j) - a_j| + \sum_{i=1}^k \min(A_i-D_i,0)$. Note that this objective is simply the discrepancy objective, plus another term which looks like discrepancy objective for categories. However, we must make an important distinction. The discrepancy objective penalizes low degree nodes both directly, and indirectly via the targets degrees for other nodes, since an extra recommendation for one node is one ``stolen'' from another node. The second term in our new objective penalizes low-degree categories, but does not necessarily penalize oversaturated categories.

\begin{theorem}
Assume the conditions of theorem 1, and let $C_1,\ldots, C_k$ be a partition of $L$, with $A_1,\ldots, A_k$ an arbitrary sequence of non-negative integers. Then a flow network exists, whose flow cost equals the objective $\sum_{v_i\in R} |\deg_H^-(v_j) - a_j| + \sum_{i=1}^k \min(A_i-D_i,0)$.
\end{theorem}
\begin{proof}
The proof is similar to that of Theorem 1. For each category $C_i$, we create three nodes: $t_1^i,t_2^i$ and $t_3^i$. The first two of these nodes are connected to the items in their category with arcs having the same costs/capacities as they did in the original network. Unlike the original network, where $t_2^i$ would have non-zero demand, both nodes in this construction have zero demand for flow. We set the demand for $t^3_i$ to be the category constraint $A_i$, and connect it to $t_2^i$ with an arc of capacity $A_i$ and cost 0.

In addition to these nodes, we create two more nodes, common to all categories: a distributor node $s_1$ and a supersink $s_2$ that are common to all the categories. The distributor $s_1$ can accept any amount of flow from a node of type $t_2^i$ at cost 1. It can also move any amount of flow to a node of type $t_3^i$ at 0 cost. Finally, the supersink $s_2$ has demand $lc-\sum_{i=1}^k A_i$, and accepts unbounded flow from nodes of type $t_3^i$ at no cost. Since we require that the category requirements sum to less than the total number of requirements, this node always has non-negative demand, and feasibility is ensured. Figure~\ref{fig:catflownetwork} illustrates the construction:

\begin{figure}
\centering
\includegraphics[width=.99\textwidth]{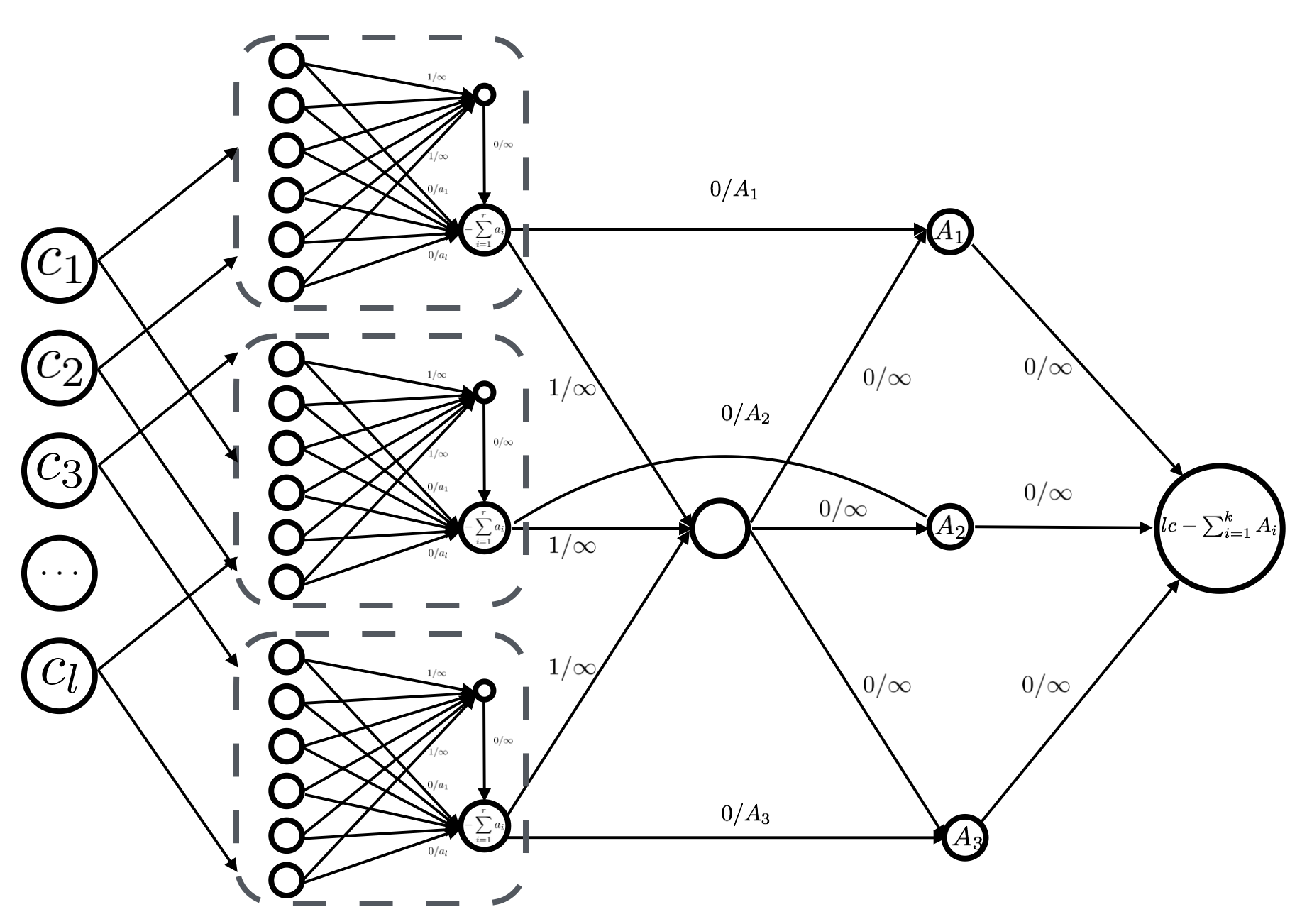}
\caption{The network flow model for category targets with nodes labelled with their supply and arcs labeled with their cost/capacity. The central node with no supply or demand is the distributor $s_1$. The rightmost node is the supersink $s_2$.}
\label{fig:catflownetwork}
\end{figure}

From the proof of Theorem 1, it follows that ignoring the cost of moving flow between nodes of type $t_2$ and $t_3$, the cost of the flow in this network is the discrepancy of the network from the target distribution. Therefore, we only need to account for the second term in our objective. The arc connecting $t_2^i$ and $t_3^i$ ensures that only $A_i$ units of flow can go uncharged towards the sink $t_3^i$. Every other unit of flow must first go to the distributor, then to $t_3^i$ to satisfy the demand of this node. Since the demand of this node is only $A_i$, we charge 1 unit of cost for every unit of flow we can't send to this category, but no cost for units of flow above the demand of this node. This gives us the term $\min(A_i-D_i,0)$ for the contribution of this category to the objective, and completes the proof.
\end{proof}
\fi

\subsection{Greedy Algorithm}
\label{subsection:Greedy2}
In this section we describe an alternative approach to solving the discrepancy minimization problem that does not require the use minimum cost flow solvers, which can be efficient in practice, but do not guarantee linear runtimes. Our greedy algorithm constructs the solution subgraph iteratively, making a discrepancy reducing recommendation whenever possible. If such an edge is not available, then we choose from all available recommendations. Our choice of recommendation is conditioned on the quality of this recommendation, as measured by our black-box relevance function $rel(u,v)$.

Since the greedy algorithm considers all discrepancy reducing recommendations for a user at the same time, a large number of candidate recommendations may lead to the greedy algorithm making subpar selections, since almost every recommendation we consider early on will likely be a discrepancy reducing edge. In order to moderate this effect, we include a parameter $q>1$ which we use to reweigh the relevance scores. The larger the $q$ is, the more our greedy algorithm prefers making relevant recommendations. On the other hand, if we pick a $q$ which is too large, then we overprioritize high relevance values, and the greedy algorithm effectively turns into the standard recommendation approach. To balance these concerns, we run the greedy algorithm with different settings of $q$, and select the solution with the highest predicted rating quality which has discrepancy at most 10\% higher than the best solution we generate.

\begin{algorithm}[h]
 \KwData{$G=(L,R,E)$, the graph of candidate recommendations, a target indegree distribution $\{a_v\}_{v=1}^r$, tuning parameter $q$, and a relevance function $rel:L\times R\to[0,1]$}
 \KwResult{Solution subgraph $H$}
 $H \longleftarrow \emptyset$\;
 \For{$j=1$ \KwTo $c$}{
  \ForEach{$u\in L$}{
   $D\longleftarrow \{v\in N_G(u): \deg_H(v) < a_v\}$\;
  \eIf{$D \not= \emptyset$}{
    Sample $e=(u,v)$ from $D$ with $p(e) \propto rel(u,v)^q$\;
   }{
    Sample $e=(u,v)$ from $N_G(u)$ with $p(e) \propto rel(u,v)^q$\;
  }
  $H \longleftarrow H \bigcup \{(u,v)\}$\;
 }
}
 \caption{The Greedy Algorithm for Discrepancy Minimization for a fixed value of $q$}
\end{algorithm} 

\section{Experiments}
\label{section:Chapter2Experiments}
In this section, we put our model to the test. Our findings are summarized below, and we discuss each point further in the following subsections.

\ifdefined\THESIS
\begin{enumerate}
\else
\begin{enumerate}[leftmargin=1em,labelwidth=*,align=left]
\fi
\item Our fast models perform well at optimizing for pre-existing notions of diversity such as aggregate diversity and the Gini index despite these measures not being explicitly referenced in our model. Conversely, we show that optimizing directly only for aggregate diversity (either by using heuristics or solving to optimality) does not yield results that are diverse by the other measures \ifdefined\THESIS(See Tables \ref{table:MLmetrics}, \ref{table:NFmetrics}, \ref{table:MLBigmetrics}).\else (See Table \ref{table:MLmetrics})\fi
\item Normalized discrepancy can often be reduced by more than 50\%, at the cost of only a 15-30\% change in average recommendation quality. Both the two-pass method, and the weighted method performed well in producing a smooth trade-off between recommendation quality and discrepancy reduction, with large gains in discrepancy being made for minimal recommendation quality loss (See Figure~\ref{fig:comp}). The two-pass method is optimal, but the weighted-model provides a good approximation of two-pass method's output with less computational overhead.
\item Sales diversity maximization problems become easier as the display constraints are relaxed since there are more opportunities for the system to make unconventional recommendations. We show that the advantage our optimization based approach has over competing approaches gets bigger as display constraints are tightened, which is desirable for applications on mobile platforms where screen real estate is scarce (see Figure ~\ref{fig:listlength}).
\item Using the uniform target distribution can lead the optimizer to pick subgraphs where degree constraints are violated by large margins at certain nodes. To remedy this, we advocate the use of target distributions that move towards resembling the underlying degree distribution rather than the uniform distribution (See Figure~\ref{fig:ConvexDegreeCDF}).
\end{enumerate}

\textbf{Experimental Setup and Datasets.} All of our experiments were conducted on a desktop computer with an Intel i5 processor clocked at 2.7GHz, and with 16GB of memory. We used three rating datasets to generate the graphs we fed to our flow solvers: MovieLens-1m, MovieLens-10m \cite{MLData} and the Netflix Prize dataset. 

We pre-process the datasets to ensure that every user and every item has an adequate amount of data on which to base predictions. This post processing leaves the MovieLens-1m data with 5800 users and 3600 items, the MovieLens-10m dataset with 67000 users and 9000 items, and the Netflix dataset with 8000 users and 5000 items. The use of these datasets is standard in the recommender systems literature. In this paper, we consider the rating data to be triples of the form $(user,item,rating)$, and discard any extra information.

\ifdefined\THESIS
\else
We will henceforth report results only using the MovieLens data due to space constraints since the results from the Netflix data are similar.
\fi
We used version 0.4.4 of the RankSys project to generate recommendations using standard collaborative filtering approaches~\cite{sandoval2015novelty} .The resulting network flow problems were optimized using a modified version of the MCFSimplex solver due to Bertolini and Frangioni~\cite{Frangioni2010MCFSimplex}. Our choice of MCFSimplex was motivated by its open-source status and efficiency, but any other minimum cost flow solver such as CPLEX or Gurobi which accepts flow problems in the standard DIMACS format can also be used by our algorithms. Our discrepancy minimization code is available at \url{}.


\textbf{Quality Evaluation.} To evaluate the quality of our method, we employ a modified version of $k$-fold cross-validation. In particular, for each user in our datasets who has an high enough number of observed ratings, we divide the rating set into 10 equal sized subsets, and place each subset in one of 10 test sets. When creating the test sets, we filter out the items which received a rating of 1 or 2 and keep the items which received a rating of 3 or higher in order to ensure the relevance of our selections. We then define the precision of our recommendation list to be the number of items we recommend among all of our top-$c$ recommendations which are also included in the test set. This provides an underestimate of the relevance of our recommendations, as there might be items which are relevant, but for which we have no record of the user liking. A simpler version of our hold-out method is utilized in other works  \cite{cooper2014random,vargas2014improving} where only a single random split is made. Using a $10$-fold split of the test data enables us to run a signed rank test, and test whether the improvements made by our algorithms are statistically significant.

Our methodology stands in contrast with the methodology used by Adomavicius and Kwon to evaluate the effectiveness of their aggregate diversity maximization framework \cite{adomavicius2011maximizing}. They use a metric called prediction-in-top-$c$, which measures the average predicted relevance of the $c$ recommended items for each user. We believe that using predicted ratings for relevance evaluation purposes is flawed since these predictions are approximate in the first place. Furthermore, using the relevance values used by the recommender make comparisons across different recommenders difficult, as each recommender has its own scale.

\textbf{Supergraph Generation.} All of our optimization problems require that a supergraph of candidate recommendations be given. For each dataset we used, we generated 240 supergraphs in total. This is the result of using 10 training sets, 4 different recommender approaches, and 6 different quality thresholds enforced by picking the top 50, 100, 200, 300, 400 and 500 recommendations for each user. We use $k$ to denote the number of candidate recommendations in the supergraph. The matrix factorization model we utilize~\cite{hu2008collaborative} comes with three parameter settings: a regularization parameter $\lambda$, a confidence parameter $\alpha$ and the number of latent factors \cite{hu2008collaborative}. 

While the authors report an $\alpha$ value of 40 is suitable for most applications, we set a lower value of $\alpha=30$ in order to obtain more diverse candidate recommendation lists. The regularization parameter $\lambda$ was tuned with cross-validation as recommended by the authors, and the model was trained with 50 latent factors.  For the neighborhood based methods, we consider neighborhoods of size 100 in both the item-based and user-based cases. For these recommenders we opted use the inverted neighborhood policy approach described in \cite{vargas2014improving} in order to obtain more diverse candidate recommendation lists. Instead of using the top 100 most similar items to an item $i$ as the neighborhood, this approach uses the items which have item $i$ in their top 100 neighborhoods. We also used Jaccard similarity in order to measure similarity between pairs of users and items in the neighborhood based methods.

The authors of the random walk recommender we implemented consider a parameter setting $\alpha$ which raises every element of the transition matrix to the power $\alpha$ and find that predictive accuracy is maximized for $\alpha=1.5$ \cite{cooper2014random}. Since they conduct their experimental validation on the same datasets as ours, we also use this parameter in our tests. In our tables, we shorten the names of these recommenders as MF for the matrix factorization model, IB and UB for the (item- and user-) neighborhood based approaches, and RW for the random walk approach.

\ifdefined\THESIS
\noindent\textbf{Trading Off Discrepancy and Rating Quality.} The higher we set the number of candidate recommendations per user, the larger the input graph $G$ will be, giving our algorithm more freedom to minimize discrepancy. On the other hand, in order to include more edges in $G$, we will have to resort to using more lower quality recommendations which will be reflected in our post-processed solution. Therefore, there is a trade-off between minimizing the discrepancy from the target distribution, and maintaining a high average recommendation quality. We will quantify this trade-off in our experiments and show that large discrepancy reductions can be made even with small compromises in recommendation quality.
\fi

\ifdefined\THESIS
\noindent\textbf{Post-Processing CF for a Diversified Recommender.} To summarize the discussion so far, we start with an arbitrary rating function which can predict the relevance of any item to a given user. We use this rating function to generate a large number of candidate recommendations which we encode in a weighted bipartite graph. We use this graph along with the designer-specified degree constraints on users and items to create a discrepancy minimization problem. Finally, we measure the quality of the resulting solution by comparing the predictive accuracy on the held-out test data.
\fi
\subsection{Comparison To Other Methods and Metrics}
\label{subsection:comparison}
In this section we compare our discrepancy minimization framework to other similar approaches. In particular, we test 6 different approaches to diversifying the recommendation lists.

\ifdefined\THESIS
\begin{itemize}
\else
\begin{itemize}[leftmargin=1em,labelwidth=*,align=left]
\fi
\item \textbf{Top(TOP)}: The standard method considers the unmodified output of the underlying recommender, and makes the top-k recommendations for each user. This is the undiversified solution but provides the highest rating quality.
\item \textbf{Two pass (GOL)}: The two-pass method first finds the lowest discrepancy value achievable with the given graph for the current target degrees, and then in a second pass, finds the highest rating solution which achieves this minimum.
\item \textbf{Aggregate Diversity (AGG)}: The aggregate diversity maximizing method is also optimized using our own flow-based framework, by running our min-cost flow algorithms with the setting of $a_i=1$ as described in Theorem~\ref{thm:aggdiv}.
\item \textbf{PC Reranking (PC), FD Reranking (FD) and Bayes Rule Reranking (AB)}: These diversifiers are due Vargas and Castells \cite{vargas2011rank,vargas2014improving}, and were discussed in detail in our related work section.
\item \textbf{Greedy (GRD)}: This is an implementation of the greedy heuristic described in Section~\ref{subsection:Greedy2}.
\end{itemize}

We evaluate these different approaches on the following metrics, all measured for the top-$n$ recommendation task on both the MovieLens and Netflix data.

\ifdefined\THESIS
\begin{itemize}
\else
\begin{itemize}[leftmargin=1em,labelwidth=*,align=left]
\fi
\item \textbf{D@n:} Discrepancy from the uniform distribution, normalized to fit in the $[0,1]$ range by dividing by the maximum discrepancy achievable, i.e. $2\sum_{i=1}^l c_i$. 
\item \textbf{A@n:} The fraction of items which received a recommendation.
\item \textbf{G@n:} The Gini index of the degree distribution of items. \ifdefined\THESIS If the degree distribution of the items is given as a sorted list $\{d_i\}_{i=1}^r$, then the Gini index is defined as follows:
\[ G = \frac{1}{r}\left( r+1-2\frac{\sum_{i=1}^r (r+1-i)d_i}{\sum_{i=1}^r d_i}\right) \]\fi
\item \textbf{E@n:} The entropy of the probability distribution formed by normalizing this degree distribution. \ifdefined\THESIS Given the same degree distribution as above, entropy is defined as follows:
\[ E = \sum_{i=1}^r -\frac{d_i}{\sum_{i=1}^r d_i} \log\left(\frac{d_i}{\sum_{i=1}^r d_i}\right)\]\fi
\item \textbf{P@n:} Precision, measured as the fraction of items in the recommendation list which are part of the test set.
\end{itemize}

\ifdefined\THESIS
Table~\ref{table:NFmetrics} summarizes our results for the Netflix dataset and Table~\ref{table:MLmetrics} does the same for the MovieLens-1m dataset.
\else
Table~\ref{table:MLmetrics} summarizes our results for the MovieLens-1m dataset. The results from the other datasets can be found in the full version of this paper.
\fi

\ifdefined\THESIS
\begin{sidewaystable}
\resizebox{\linewidth}{!}{%
\begin{tabular}{|c|r|ccccc?ccccc?ccccc|}
\cline{3-17}
     \multicolumn{2}{c|}{} & \multicolumn{5}{c?}{k=50} & \multicolumn{5}{c?}{k=250} & \multicolumn{5}{c|}{k=500}\\
\cline{3-17}
     \multicolumn{2}{c|}{} & P@10  & A@10  & G@10  & E@10  & D@10  & P@10  & A@10  & G@10  & E@10  & D@10 & P@10  & A@10  & G@10  & E@10  & D@10 \\
\hline
\multirow{7}{*}{MF}
& TOP & \textbf{0.433} & 0.263 & 0.924 & 5.974 & 0.823 & \textbf{0.433} & 0.263 & 0.924 & 5.974 & 0.823 & \textbf{0.433} & 0.263 & 0.924 & 5.974 & 0.823\\
& AGG & \emph{\textbf{0.433}} & \textbf{0.429} & 0.919 & 6.024 & 0.816 & \emph{0.432} & \textbf{0.651} & 0.906 & 6.094 & 0.804 & \emph{0.432} & \textbf{0.758} & 0.896 & 6.146 & 0.796\\
& FD & 0.340 & 0.385 & 0.826 & 6.819 & 0.712 & 0.316 & 0.478 & 0.792 & 7.000 & 0.658 & 0.330 & 0.528 & 0.802 & 6.952 & 0.659\\
& PC & 0.362 & 0.331 & 0.858 & 6.619 & 0.754 & 0.335 & 0.347 & 0.843 & 6.715 & 0.736 & 0.321 & 0.356 & 0.836 & 6.756 & 0.728\\
& AB & 0.257 & \textbf{0.429} & \textbf{0.767} & \textbf{7.081} & 0.661 & 0.167 & 0.649 & 0.648 & \textbf{7.504} & 0.503 & 0.204 & 0.735 & 0.660 & 7.472 & 0.507\\
& GOL & 0.414 & 0.423 & 0.855 & 6.451 & \textbf{0.658} & 0.358 & 0.639 & \textbf{0.646} & 7.189 & \textbf{0.414} & 0.331 & 0.754 & \textbf{0.483} & \textbf{7.581} & \textbf{0.282}\\
& GRD & 0.269 & 0.418 & 0.815 & 6.863 & 0.670 & 0.125 & 0.620 & 0.603 & 7.583 & 0.449 & 0.083 & 0.734 & 0.455 & 7.852 & 0.324\\
\hline

\multirow{7}{*}{IB}
& TOP & \textbf{0.410} & 0.257 & 0.953 & 5.450 & 0.861 & \textbf{0.410} & 0.257 & 0.953 & 5.450 & 0.861 & \textbf{0.410} & 0.257 & 0.953 & 5.450 & 0.861\\
& AGG & \emph{0.409} & \textbf{0.431} & 0.947 & 5.496 & 0.855 & \emph{0.409} & \textbf{0.653} & 0.933 & 5.573 & 0.843 & \emph{0.409} & \textbf{0.813} & 0.918 & 5.640 & 0.833\\
& FD & 0.366 & 0.352 & 0.906 & 6.202 & 0.795 & 0.321 & 0.465 & 0.842 & 6.730 & 0.708 & 0.296 & 0.633 & 0.781 & 7.035 & 0.622\\
& PC & 0.368 & 0.317 & 0.921 & 6.028 & 0.819 & 0.338 & 0.332 & 0.902 & 6.261 & 0.796 & 0.323 & 0.348 & 0.889 & 6.383 & 0.779\\
& AB & 0.282 & 0.429 & \textbf{0.861} & \textbf{6.610} & 0.742 & 0.196 & 0.641 & \textbf{0.735} & \textbf{7.231} & 0.591 & 0.214 & 0.782 & 0.695 & \textbf{7.299} & 0.528\\
& GOL & \emph{0.404} & 0.380 & 0.911 & 5.834 & \textbf{0.735} & 0.380 & 0.573 & 0.763 & 6.611 & \textbf{0.524} & 0.358 & 0.733 & \textbf{0.605} & 7.160 & \textbf{0.375}\\
& GRD & 0.258 & 0.383 & 0.898 & 6.215 & 0.753 & 0.135 & 0.556 & 0.747 & 7.163 & 0.568 & 0.090 & 0.695 & 0.610 & 7.577 & 0.437\\
\hline

\multirow{7}{*}{UB}
& TOP & \textbf{0.412} & 0.146 & 0.971 & 5.002 & 0.903 & \textbf{0.412} & 0.146 & 0.971 & 5.002 & 0.903 & \textbf{0.412} & 0.146 & 0.971 & 5.002 & 0.903\\
& AGG & \emph{\textbf{0.412}} & \textbf{0.308} & 0.967 & 5.062 & 0.895 & \emph{\textbf{0.412}} & \textbf{0.585} & 0.952 & 5.176 & 0.878 & \emph{\textbf{0.412}} & \textbf{0.777} & 0.936 & 5.262 & 0.866\\
& FD & 0.380 & 0.282 & 0.902 & 6.242 & 0.797 & 0.335 & 0.498 & 0.840 & 6.674 & 0.686 & 0.333 & 0.619 & 0.833 & 6.630 & 0.665\\
& PC & 0.391 & 0.240 & 0.922 & 6.029 & 0.833 & 0.360 & 0.275 & 0.903 & 6.252 & 0.809 & 0.350 & 0.279 & 0.898 & 6.300 & 0.802\\
& AB & 0.291 & 0.307 & \textbf{0.861} & \textbf{6.592} & 0.761 & 0.195 & 0.566 & 0.747 & \textbf{7.158} & 0.593 & 0.210 & 0.696 & 0.703 & 7.240 & 0.528\\
& GOL & \textbf{0.412} & 0.306 & 0.928 & 5.530 & \textbf{0.760} & 0.382 & 0.582 & \textbf{0.732} & 6.618 & \textbf{0.486} & 0.342 & 0.772 & \textbf{0.507} & \textbf{7.358} & \textbf{0.297}\\
& GRD & 0.265 & 0.303 & 0.900 & 6.210 & 0.766 & 0.136 & 0.566 & 0.703 & 7.302 & 0.519 & 0.091 & 0.746 & 0.519 & 7.757 & 0.353\\
\hline

\multirow{7}{*}{RW}
& TOP & 0.303 & 0.072 & 0.992 & 3.710 & 0.970 & 0.303 & 0.072 & 0.992 & 3.710 & 0.970 & 0.303 & 0.072 & 0.992 & 3.710 & 0.970\\
& AGG & \emph{0.304} & \textbf{0.262} & 0.989 & 3.786 & 0.960 & \emph{0.302} & \textbf{0.515} & 0.976 & 3.909 & 0.943 & 0.297 & \textbf{0.670} & 0.967 & 3.965 & 0.936\\
& FD & 0.343 & 0.232 & 0.964 & 5.203 & 0.898 & 0.289 & 0.466 & 0.911 & 5.968 & 0.780 & 0.284 & 0.523 & 0.908 & 5.792 & 0.761\\
& PC & \textbf{0.351} & 0.205 & 0.967 & 5.147 & 0.905 & \textbf{0.339} & 0.364 & 0.933 & 5.819 & 0.829 & \textbf{0.348} & 0.404 & 0.929 & 5.842 & 0.819\\
& AB & 0.287 & 0.261 & \textbf{0.957} & \textbf{5.409} & 0.889 & 0.209 & 0.488 & \textbf{0.881} & \textbf{6.429} & 0.766 & 0.183 & 0.550 & 0.846 & \textbf{6.625} & 0.705\\
& GOL & \emph{0.319} & 0.228 & 0.980 & 4.122 & \textbf{0.887} & \emph{0.302} & 0.489 & 0.896 & 5.177 & \textbf{0.689} & 0.248 & 0.633 & \textbf{0.780} & 5.980 & \textbf{0.536}\\
& GRD & 0.213 & 0.238 & 0.966 & 5.129 & 0.889 & 0.130 & 0.478 & 0.856 & 6.627 & 0.708 & 0.094 & 0.613 & 0.743 & 7.206 & 0.574\\
\hline
\end{tabular}
}
\caption{Comparison of different diversifiers systems on various diversification metrics for the 10 item recommendation task. The underlying dataset is the MovieLens-1m dataset and the candidate recommendations were generated using Matrix Factorization (MF), Item-Based (IB), User-Based (UB), and Random Walk Recommenders (RW). The bolded entries show the best values in each metric (ignoring the greedy algorithm), and italicized values show a statistically insignificant change from the baseline with $p<0.01$.}
\label{table:MLmetrics}
\end{sidewaystable}

\begin{sidewaystable}
\centering
\resizebox{\linewidth}{!}{%
\begin{tabular}{|c|r|ccccc?ccccc?ccccc|}
\cline{3-17}
     \multicolumn{2}{c|}{} & \multicolumn{5}{c?}{k=50} & \multicolumn{5}{c?}{k=250} & \multicolumn{5}{c|}{k=500}\\
\cline{3-17}
     \multicolumn{2}{c|}{} & P@10  & A@10  & G@10  & E@10  & D@10  & P@10  & A@10  & G@10  & E@10  & D@10 & P@10  & A@10  & G@10  & E@10  & D@10 \\
\hline
\multirow{7}{*}{MF}
& TOP & \textbf{0.734} & 0.364 & 0.893 & 6.311 & 0.771 & \textbf{0.734} & 0.364 & 0.893 & 6.311 & 0.771 & \textbf{0.734} & 0.364 & 0.893 & 6.311 & 0.771\\
& AGG & 0.733 & \textbf{0.515} & 0.888 & 6.351 & 0.765 & 0.730 & \textbf{0.743} & 0.873 & 6.426 & 0.754 & 0.726 & \textbf{0.870} & 0.863 & 6.459 & 0.748\\
& FD & 0.568 & 0.476 & 0.803 & 6.952 & 0.665 & 0.448 & 0.571 & 0.757 & 7.162 & 0.607 & 0.442 & 0.634 & 0.736 & 7.246 & 0.583\\
& PC & 0.589 & 0.432 & 0.826 & 6.829 & 0.695 & 0.510 & 0.442 & 0.818 & 6.874 & 0.685 & 0.475 & 0.453 & 0.815 & 6.889 & 0.681\\
& AB & 0.472 & 0.494 & \textbf{0.765} & \textbf{7.122} & 0.649 & 0.215 & 0.701 & 0.665 & \textbf{7.480} & 0.524 & 0.200 & 0.836 & 0.625 & 7.594 & 0.475\\
& GOL & 0.694 & 0.503 & 0.812 & 6.759 & \textbf{0.611} & 0.554 & 0.722 & \textbf{0.578} & 7.452 & \textbf{0.360} & 0.440 & 0.860 & \textbf{0.359} & \textbf{7.870} & \textbf{0.199}\\
& GRD & 0.525 & 0.495 & 0.785 & 7.031 & 0.626 & 0.243 & 0.697 & 0.566 & 7.688 & 0.411 & 0.150 & 0.826 & 0.388 & 7.975 & 0.267\\
\hline

\multirow{7}{*}{IB}
& TOP & 0.681 & 0.079 & 0.984 & 4.436 & 0.954 & 0.681 & 0.079 & 0.984 & 4.436 & 0.954 & \textbf{0.681} & 0.079 & 0.984 & 4.436 & 0.954\\
& AGG & \emph{0.682} & \textbf{0.133} & 0.984 & 4.455 & 0.952 & \textbf{0.682} & 0.216 & 0.982 & 4.485 & 0.948 & \textbf{0.681} & 0.219 & 0.982 & 4.486 & 0.948\\
& FD & 0.650 & 0.127 & \textbf{0.948} & \textbf{5.584} & 0.912 & 0.454 & 0.265 & \textbf{0.881} & \textbf{6.437} & 0.799 & 0.281 & 0.476 & 0.842 & \textbf{6.712} & 0.697\\
& PC & 0.650 & 0.124 & 0.950 & 5.547 & 0.913 & 0.532 & 0.212 & 0.912 & 6.143 & 0.843 & 0.474 & 0.303 & 0.897 & 6.321 & 0.808\\
& AB & 0.532 & 0.125 & 0.967 & 5.177 & 0.924 & 0.449 & 0.196 & 0.950 & 5.586 & 0.885 & 0.484 & 0.329 & 0.927 & 5.751 & 0.811\\
& GOL & \textbf{0.685} & 0.130 & 0.979 & 4.635 & \textbf{0.908} & 0.670 & \textbf{0.276} & 0.943 & 5.207 & \textbf{0.782} & 0.620 & \textbf{0.488} & \textbf{0.834} & 6.041 & \textbf{0.600}\\
& GRD & 0.476 & 0.130 & 0.955 & 5.460 & 0.909 & 0.263 & 0.272 & 0.865 & 6.530 & 0.786 & 0.201 & 0.480 & 0.750 & 7.167 & 0.621\\
\hline

\multirow{7}{*}{UB}
& TOP & 0.701 & 0.101 & 0.976 & 4.863 & 0.931 & 0.701 & 0.101 & 0.976 & 4.863 & 0.931 & 0.701 & 0.101 & 0.976 & 4.863 & 0.931\\
& AGG & \emph{0.702} & \textbf{0.192} & 0.975 & 4.896 & 0.926 & \emph{\textbf{0.703}} & \textbf{0.433} & 0.974 & 4.915 & 0.924 & \emph{\textbf{0.702}} & \textbf{0.652} & 0.974 & 4.915 & 0.924\\
& FD & 0.631 & 0.190 & 0.936 & 5.830 & 0.860 & 0.456 & 0.385 & 0.868 & 6.543 & 0.740 & 0.380 & 0.552 & 0.835 & 6.721 & 0.677\\
& PC & 0.632 & 0.186 & 0.938 & 5.805 & 0.864 & 0.538 & 0.282 & 0.902 & 6.262 & 0.807 & 0.506 & 0.302 & 0.894 & 6.339 & 0.792\\
& AB & 0.561 & 0.184 & \textbf{0.933} & \textbf{5.888} & 0.867 & 0.380 & 0.367 & \textbf{0.866} & \textbf{6.565} & 0.752 & 0.385 & 0.505 & 0.841 & 6.620 & 0.685\\
& GOL & \emph{\textbf{0.705}} & \textbf{0.192} & 0.961 & 5.183 & \textbf{0.851} & 0.658 & 0.419 & 0.873 & 5.963 & \textbf{0.659} & 0.563 & 0.630 & \textbf{0.710} & \textbf{6.728} & \textbf{0.466}\\
& GRD & 0.522 & 0.192 & 0.933 & 5.897 & 0.851 & 0.284 & 0.409 & 0.804 & 6.944 & 0.671 & 0.193 & 0.602 & 0.656 & 7.477 & 0.500\\
\hline

\multirow{7}{*}{RW}
& TOP & 0.615 & 0.027 & 0.990 & 3.901 & 0.979 & \textbf{0.615} & 0.027 & 0.990 & 3.901 & 0.979 & \textbf{0.615} & 0.027 & 0.990 & 3.901 & 0.979\\
& AGG & 0.615 & \textbf{0.069} & 0.990 & 3.914 & 0.977 & \textbf{0.615} & \textbf{0.194} & 0.988 & 3.960 & 0.971 & 0.612 & \textbf{0.329} & 0.985 & 4.008 & 0.963\\
& FD & \emph{0.659} & 0.068 & \textbf{0.966} & 5.156 & \textbf{0.941} & 0.496 & 0.154 & 0.954 & 5.486 & 0.889 & 0.489 & 0.187 & 0.958 & 5.327 & 0.877\\
& PC & \textbf{0.661} & 0.067 & \textbf{0.966} & \textbf{5.157} & 0.942 & 0.544 & 0.179 & \textbf{0.930} & \textbf{5.928} & 0.857 & 0.530 & 0.267 & 0.909 & \textbf{6.190} & 0.809\\
& AB & 0.525 & 0.055 & 0.980 & 4.652 & 0.957 & 0.462 & 0.094 & 0.972 & 4.955 & 0.928 & 0.491 & 0.137 & 0.970 & 4.898 & 0.901\\
& GOL & \emph{0.622} & \textbf{0.069} & 0.988 & 4.042 & \textbf{0.941} & 0.611 & 0.193 & 0.963 & 4.529 & \textbf{0.823} & 0.568 & 0.327 & \textbf{0.908} & 5.125 & \textbf{0.703}\\
& GRD & 0.439 & 0.068 & 0.965 & 5.158 & 0.941 & 0.249 & 0.191 & 0.880 & 6.366 & 0.825 & 0.182 & 0.319 & 0.797 & 6.907 & 0.713\\
\hline
\end{tabular}
}
\caption{Comparison of different diversifiers systems on various diversification metrics for the 10 item recommendation task. The underlying dataset is the Netflix dataset and the candidate recommendations were generated using Matrix Factorization (MF), Item-Based (IB), User-Based (UB), and Random Walk Recommenders (RW). The bolded entries show the best values in each metric (ignoring the greedy algorithm), and italicized values show a statistically insignificant change from the baseline with $p<0.01$.}
\label{table:NFmetrics}
\end{sidewaystable}

\begin{sidewaystable}
\centering
\resizebox{\linewidth}{!}{%
\begin{tabular}{|c|r|ccccc?ccccc?ccccc|}
\cline{3-17}
     \multicolumn{2}{c|}{} & \multicolumn{5}{c?}{k=50} & \multicolumn{5}{c?}{k=250} & \multicolumn{5}{c|}{k=500}\\
\cline{3-17}
     \multicolumn{2}{c|}{} & P@10  & A@10  & G@10  & E@10  & D@10  & P@10  & A@10  & G@10  & E@10  & D@10 & P@10  & A@10  & G@10  & E@10  & D@10 \\
\hline
\multirow{7}{*}{MF}
& TOP & \textbf{0.451} & 0.817 & 0.959 & 6.063 & 0.901 & \textbf{0.451} & 0.817 & 0.959 & 6.063 & 0.901 & \textbf{0.451} & 0.817 & 0.959 & 6.063 & 0.901\\
& AGG & \emph{\textbf{0.451}} & \textbf{0.849} & 0.955 & 6.094 & 0.881 & 0.447 & \textbf{0.890} & 0.947 & 6.124 & 0.873 & 0.446 & \textbf{0.913} & 0.939 & 6.147 & 0.868\\
& FD & 0.361 & 0.831 & 0.925 & 6.881 & 0.852 & 0.355 & 0.840 & 0.911 & 7.069 & 0.820 & 0.374 & 0.849 & 0.921 & 6.912 & 0.828\\
& PC & 0.382 & 0.823 & 0.939 & 6.639 & 0.876 & 0.353 & 0.823 & 0.933 & 6.747 & 0.868 & 0.341 & 0.819 & 0.930 & 6.798 & 0.864\\
& AB & 0.267 & 0.837 & \textbf{0.899} & \textbf{7.213} & 0.829 & 0.188 & 0.867 & \textbf{0.837} & \textbf{7.681} & 0.710 & 0.205 & 0.887 & 0.834 & \textbf{7.655} & 0.682\\
& GOL & \emph{0.439} & 0.835 & 0.919 & 6.486 & \textbf{0.822} & 0.417 & 0.851 & 0.853 & 7.224 & \textbf{0.645} & 0.382 & 0.860 & \textbf{0.792} & 7.541 & \textbf{0.556}\\
& GRD & 0.275 & 0.834 & 0.926 & 6.837 & 0.832 & 0.128 & 0.859 & 0.847 & 7.653 & 0.707 & 0.087 & 0.880 & 0.775 & 8.041 & 0.617\\
\hline

\multirow{7}{*}{IB}
& TOP & \textbf{0.388} & 0.828 & 0.968 & 5.543 & 0.910 & \textbf{0.388} & 0.828 & 0.968 & 5.543 & 0.910 & \textbf{0.388} & 0.828 & 0.968 & 5.543 & 0.910\\
& AGG & \emph{0.387} & \textbf{0.876} & 0.966 & 5.615 & 0.905 & \emph{0.386} & \textbf{0.932} & 0.959 & 5.654 & 0.903 & \emph{0.386} & \textbf{0.970} & 0.951 & 5.684 & 0.899\\
& FD & 0.359 & 0.844 & 0.953 & 6.160 & 0.879 & 0.338 & 0.873 & 0.920 & 6.842 & 0.812 & 0.306 & 0.911 & 0.898 & 6.938 & 0.759\\
& PC & 0.358 & 0.838 & 0.959 & 5.972 & 0.891 & 0.331 & 0.842 & 0.953 & 6.198 & 0.882 & 0.321 & 0.847 & 0.949 & 6.306 & 0.876\\
& AB & 0.275 & 0.859 & \textbf{0.939} & \textbf{6.549} & 0.856 & 0.232 & 0.915 & \textbf{0.866} & \textbf{7.412} & 0.745 & 0.231 & 0.960 & 0.823 & \textbf{7.557} & 0.661\\
& GOL & 0.362 & 0.844 & 0.941 & 5.934 & \textbf{0.853} & 0.354 & 0.901 & 0.897 & 6.873 & \textbf{0.659} & 0.335 & 0.949 & \textbf{0.820} & 7.301 & \textbf{0.580}\\
& GRD & 0.239 & 0.844 & 0.950 & 6.257 & 0.863 & 0.131 & 0.880 & 0.892 & 7.206 & 0.749 & 0.095 & 0.915 & 0.819 & 7.744 & 0.640\\
\hline

\multirow{7}{*}{UB}
& TOP & \textbf{0.440} & 0.802 & 0.970 & 5.483 & 0.919 & \textbf{0.440} & 0.802 & 0.970 & 5.483 & 0.919 & \textbf{0.440} & 0.802 & 0.970 & 5.483 & 0.919\\
& AGG & \emph{\textbf{0.440}} & \textbf{0.853} & 0.968 & 5.619 & 0.919 & \emph{0.439} & \textbf{0.927} & 0.954 & 5.681 & 0.909 & 0.437 & \textbf{0.978} & 0.941 & 5.732 & 0.893\\
& FD & 0.387 & 0.828 & 0.938 & 6.596 & 0.851 & 0.344 & 0.904 & 0.886 & 7.166 & 0.749 & 0.337 & 0.934 & 0.867 & 7.201 & 0.718\\
& PC & 0.399 & 0.817 & 0.951 & 6.308 & 0.879 & 0.373 & 0.838 & 0.943 & 6.494 & 0.867 & 0.365 & 0.832 & 0.942 & 6.539 & 0.864\\
& AB & 0.317 & 0.832 & \textbf{0.922} & \textbf{6.918} & 0.834 & 0.251 & 0.912 & \textbf{0.827} & \textbf{7.594} & 0.666 & 0.237 & 0.951 & 0.737 & \textbf{7.899} & 0.548\\
& GOL & 0.412 & 0.842 & 0.930 & 6.145 & \textbf{0.833} & 0.398 & 0.916 & 0.845 & 7.185 & \textbf{0.592} & 0.372 & 0.939 & \textbf{0.723} & 7.621 & \textbf{0.467}\\
& GRD & 0.279 & 0.831 & 0.940 & 6.472 & 0.835 & 0.154 & 0.907 & 0.836 & 7.553 & 0.652 & 0.115 & 0.942 & 0.716 & 8.087 & 0.509\\
\hline
\end{tabular}
}
\caption{Comparison of different diversifiers systems on various diversification metrics for the 10 item recommendation task. The underlying dataset is the MovieLens-10m dataset and the candidate recommendations were generated using Matrix Factorization (MF), Item-Based (IB), User-Based (UB) recommenders. The bolded entries show the best values in each metric (ignoring the greedy algorithm), and italicized values show a statistically insignificant change from the baseline with $p<0.01$.}
\label{table:MLBigmetrics}
\end{sidewaystable}

\else

\begin{table*}
\resizebox{\textwidth}{!}{%
\begin{tabular}{|c|r|ccccc?ccccc?ccccc|}
\cline{3-17}
     \multicolumn{2}{c|}{} & \multicolumn{5}{c?}{k=50} & \multicolumn{5}{c?}{k=250} & \multicolumn{5}{c|}{k=500}\\
\cline{3-17}
     \multicolumn{2}{c|}{} & P@10  & A@10  & G@10  & E@10  & D@10  & P@10  & A@10  & G@10  & E@10  & D@10 & P@10  & A@10  & G@10  & E@10  & D@10 \\
\hline
\multirow{7}{*}{MF}
& TOP & \textbf{0.433} & 0.263 & 0.924 & 5.974 & 0.823 & \textbf{0.433} & 0.263 & 0.924 & 5.974 & 0.823 & \textbf{0.433} & 0.263 & 0.924 & 5.974 & 0.823\\
& AGG & \emph{\textbf{0.433}} & \textbf{0.429} & 0.919 & 6.024 & 0.816 & \emph{0.432} & \textbf{0.651} & 0.906 & 6.094 & 0.804 & \emph{0.432} & \textbf{0.758} & 0.896 & 6.146 & 0.796\\
& FD & 0.340 & 0.385 & 0.826 & 6.819 & 0.712 & 0.316 & 0.478 & 0.792 & 7.000 & 0.658 & 0.330 & 0.528 & 0.802 & 6.952 & 0.659\\
& PC & 0.362 & 0.331 & 0.858 & 6.619 & 0.754 & 0.335 & 0.347 & 0.843 & 6.715 & 0.736 & 0.321 & 0.356 & 0.836 & 6.756 & 0.728\\
& AB & 0.257 & \textbf{0.429} & \textbf{0.767} & \textbf{7.081} & 0.661 & 0.167 & 0.649 & 0.648 & \textbf{7.504} & 0.503 & 0.204 & 0.735 & 0.660 & 7.472 & 0.507\\
& GOL & 0.414 & 0.423 & 0.855 & 6.451 & \textbf{0.658} & 0.358 & 0.639 & \textbf{0.646} & 7.189 & \textbf{0.414} & 0.331 & 0.754 & \textbf{0.483} & \textbf{7.581} & \textbf{0.282}\\
& GRD & 0.269 & 0.418 & 0.815 & 6.863 & 0.670 & 0.125 & 0.620 & 0.603 & 7.583 & 0.449 & 0.083 & 0.734 & 0.455 & 7.852 & 0.324\\
\hline

\multirow{7}{*}{IB}
& TOP & \textbf{0.410} & 0.257 & 0.953 & 5.450 & 0.861 & \textbf{0.410} & 0.257 & 0.953 & 5.450 & 0.861 & \textbf{0.410} & 0.257 & 0.953 & 5.450 & 0.861\\
& AGG & \emph{0.409} & \textbf{0.431} & 0.947 & 5.496 & 0.855 & \emph{0.409} & \textbf{0.653} & 0.933 & 5.573 & 0.843 & \emph{0.409} & \textbf{0.813} & 0.918 & 5.640 & 0.833\\
& FD & 0.366 & 0.352 & 0.906 & 6.202 & 0.795 & 0.321 & 0.465 & 0.842 & 6.730 & 0.708 & 0.296 & 0.633 & 0.781 & 7.035 & 0.622\\
& PC & 0.368 & 0.317 & 0.921 & 6.028 & 0.819 & 0.338 & 0.332 & 0.902 & 6.261 & 0.796 & 0.323 & 0.348 & 0.889 & 6.383 & 0.779\\
& AB & 0.282 & 0.429 & \textbf{0.861} & \textbf{6.610} & 0.742 & 0.196 & 0.641 & \textbf{0.735} & \textbf{7.231} & 0.591 & 0.214 & 0.782 & 0.695 & \textbf{7.299} & 0.528\\
& GOL & \emph{0.404} & 0.380 & 0.911 & 5.834 & \textbf{0.735} & 0.380 & 0.573 & 0.763 & 6.611 & \textbf{0.524} & 0.358 & 0.733 & \textbf{0.605} & 7.160 & \textbf{0.375}\\
& GRD & 0.258 & 0.383 & 0.898 & 6.215 & 0.753 & 0.135 & 0.556 & 0.747 & 7.163 & 0.568 & 0.090 & 0.695 & 0.610 & 7.577 & 0.437\\
\hline

\multirow{7}{*}{UB}
& TOP & \textbf{0.412} & 0.146 & 0.971 & 5.002 & 0.903 & \textbf{0.412} & 0.146 & 0.971 & 5.002 & 0.903 & \textbf{0.412} & 0.146 & 0.971 & 5.002 & 0.903\\
& AGG & \emph{\textbf{0.412}} & \textbf{0.308} & 0.967 & 5.062 & 0.895 & \emph{\textbf{0.412}} & \textbf{0.585} & 0.952 & 5.176 & 0.878 & \emph{\textbf{0.412}} & \textbf{0.777} & 0.936 & 5.262 & 0.866\\
& FD & 0.380 & 0.282 & 0.902 & 6.242 & 0.797 & 0.335 & 0.498 & 0.840 & 6.674 & 0.686 & 0.333 & 0.619 & 0.833 & 6.630 & 0.665\\
& PC & 0.391 & 0.240 & 0.922 & 6.029 & 0.833 & 0.360 & 0.275 & 0.903 & 6.252 & 0.809 & 0.350 & 0.279 & 0.898 & 6.300 & 0.802\\
& AB & 0.291 & 0.307 & \textbf{0.861} & \textbf{6.592} & 0.761 & 0.195 & 0.566 & 0.747 & \textbf{7.158} & 0.593 & 0.210 & 0.696 & 0.703 & 7.240 & 0.528\\
& GOL & \textbf{0.412} & 0.306 & 0.928 & 5.530 & \textbf{0.760} & 0.382 & 0.582 & \textbf{0.732} & 6.618 & \textbf{0.486} & 0.342 & 0.772 & \textbf{0.507} & \textbf{7.358} & \textbf{0.297}\\
& GRD & 0.265 & 0.303 & 0.900 & 6.210 & 0.766 & 0.136 & 0.566 & 0.703 & 7.302 & 0.519 & 0.091 & 0.746 & 0.519 & 7.757 & 0.353\\
\hline

\multirow{7}{*}{RW}
& TOP & 0.303 & 0.072 & 0.992 & 3.710 & 0.970 & 0.303 & 0.072 & 0.992 & 3.710 & 0.970 & 0.303 & 0.072 & 0.992 & 3.710 & 0.970\\
& AGG & \emph{0.304} & \textbf{0.262} & 0.989 & 3.786 & 0.960 & \emph{0.302} & \textbf{0.515} & 0.976 & 3.909 & 0.943 & 0.297 & \textbf{0.670} & 0.967 & 3.965 & 0.936\\
& FD & 0.343 & 0.232 & 0.964 & 5.203 & 0.898 & 0.289 & 0.466 & 0.911 & 5.968 & 0.780 & 0.284 & 0.523 & 0.908 & 5.792 & 0.761\\
& PC & \textbf{0.351} & 0.205 & 0.967 & 5.147 & 0.905 & \textbf{0.339} & 0.364 & 0.933 & 5.819 & 0.829 & \textbf{0.348} & 0.404 & 0.929 & 5.842 & 0.819\\
& AB & 0.287 & 0.261 & \textbf{0.957} & \textbf{5.409} & 0.889 & 0.209 & 0.488 & \textbf{0.881} & \textbf{6.429} & 0.766 & 0.183 & 0.550 & 0.846 & \textbf{6.625} & 0.705\\
& GOL & \emph{0.319} & 0.228 & 0.980 & 4.122 & \textbf{0.887} & \emph{0.302} & 0.489 & 0.896 & 5.177 & \textbf{0.689} & 0.248 & 0.633 & \textbf{0.780} & 5.980 & \textbf{0.536}\\
& GRD & 0.213 & 0.238 & 0.966 & 5.129 & 0.889 & 0.130 & 0.478 & 0.856 & 6.627 & 0.708 & 0.094 & 0.613 & 0.743 & 7.206 & 0.574\\
\hline
\end{tabular}
}
\caption{Comparison of different diversifiers systems on various diversification metrics for the 10 item recommendation task. The underlying dataset is the MovieLens-1m dataset and the candidate recommendations were generated using Matrix Factorization (MF), Item-Based (IB), User-Based (UB), and Random Walk Recommenders (RW). The bolded entries show the best values in each metric (ignoring the greedy algorithm), and italicized values show a statistically insignificant change from the baseline with $p<0.01$.}
\label{table:MLmetrics}
\vspace{-.8cm}
\end{table*}
\fi

\ifdefined\THESIS We start our discussion with the results on the two medium sized datasets. \fi The first thing to notice in \ifdefined\THESIS these tables \else this table \fi is that the undiversified recommendation lists perform very poorly with respect to all distributional measures. This is true with respect to even the simplest measure, aggregate diversity. The Random Walk Recommender in particular surfaces only \ifdefined\THESIS 3\% of the items in the Netflix catalog and \fi 7\% of the items in the MovieLens catalog. Other recommenders do not do particularly better, and only surface 15-20\% of items to the users via top-10 recommendation lists.

\ifdefined\THESIS
Interestingly, the standard recommendation approach is not always the most precise among the methods we tested. Both our discrepancy minimization framework, and the PC and FD reranking methods beat the prevision of the baseline recommendation approach in multiple instances, most notably for the Random Walk Recommender. This observation may seem paradoxical at first, but was in fact noted by other authors in multiple different settings \cite{vargas2014improving}. The Random Walk Recommender is notable among the methods we tested for being the recommendation approach with the highest amount of popularity bias. Therefore, our results lend support to the idea that recommender systems benefit from being optimized for metrics other than accuracy and popularity.
\fi

We also note that an optimization based approach using aggregate diversity as the objective function is unlikely to make a large difference in the degree distribution of the underlying recommender system. Indeed, for all the recommenders and for both of the datasets, the aggregate diversity maximization approach only makes significant improvements with respect to its own metric. This is to be expected from a crude measure of system-level diversity such as aggregate diversity, as it is possible to have a very lopsided recommendation distribution even while achieving full coverage of the item catalog. \ifdefined\THESIS As an extreme example, consider a recommender which seeks to make 2 recommendations per user, and suppose there are as many users as there are items. If the recommender recommends one unique item to each users, and one of the items to every user, aggregate diversity will be maximized. However, the Gini index will be high and the entropy of the degree distribution will be low. \fi Aggregate diversity is still a valuable measure of the coverage properties of the recommender, as our baseline approaches do not achieve good values for even this simple metric. However, all of the other approaches we tested also make significant gains in aggregate diversity while also significantly impacting the other metrics at the same time. In particular, the Bayes Rule reranking approach and our own discrepancy minimization almost maximize aggregate diversity on the MovieLens-1m and Netflix datasets respectively. Therefore, we believe that the best approach to aggregate diversity maximization is to optimize it by proxy, by optimizing for a more refined measure of recommendation diversity.

As expected, our method performs the best with respect to minimizing discrepancy from the uniform distribution. However, the other methods we tested also make gains towards minimizing discrepancy, which lends credibility to the use of discrepancy from the uniform distribution as a metric for evaluating the system-wide diversity of recommendations.

Among the methods we tested, the most aggressive diversifier is the Bayes Rule reranker, which obtains the best scores with respect to every metric, in almost every instance. Among the metrics we tested, our method is best suited for optimizing for aggregate diversity, the Gini index, and entropy of the degree distribution, in that order. Even though the two-pass method does not always obtain the best diversity improvements, it does so at a very low predictive cost. Furthermore, as more edges are included in the list of permissible recommendations (the second column of \ifdefined\THESIS Tables \ref{table:MLmetrics},\ref{table:NFmetrics},\ref{table:MLBigmetrics}\else Table \ref{table:MLmetrics}\fi), the advantage of the Bayes Rule reranker almost vanishes with respect to the Gini and aggregate diversity metrics, despite our recommender surfacing between 30\% to 120\% more relevant recommendations. The predictive accuracy of the \PC and \FD reranking methods fall off at a slower rate than the Bayes Rule reranking method, but our method dominates those two methods in every metric except for entropy. When 500 candidate recommendations are included for each user, our recommender creates solutions with better diversity values and better precision values than all the other recommenders we tested, even when they were supplied with more accurate and shorter candidate recommendation lists.

Finally, we comment on the performance of the greedy discrepancy reducing technique. With a relatively small number of candidate recommendations per user like $k=50$, the greedy algorithm can strike a good balance between discrepancy reduction and maintaining the relevance of recommendations. However, in this regime, the greedy diversifier is matched on all diversity metrics by the corresponding reranking based techniques, which also beat it in predictive accuracy. The greedy algorithm is much quicker to improve diversity as more candidate recommendations are added; However, in this regime the Bayes Rule Reranker matches it well in diversity measures while beating it in predictive accuracy. Since the predictive losses are too great, the greedy algorithm is an inadequate replacement for the full flow based methods.

\ifdefined\THESIS
\subsubsection{Large Dataset}
\label{subsection:largedataset}
The large dataset we tested deserves a special discussion for two reasons. First, the MovieLens-10m dataset has 7 times as many users as it has items in its catalog, compared to only a 1.5 times ratio in the MovieLens-1m dataset. With 70 times as many recommendations as items, even the standard recommendation approach achieved 80\% coverage of the item catalog. This renders aggregate diversity maximization ineffective as a method for increasing the equitability of the recommendation distribution. Second, the large dataset is too large to solve in one batch by our flow methods. Therefore for these experiments, we split the user base into 4 batches, find the minimum discrepancy graph in each of these problems and combine them. Naturally, this leads to a higher discrepancy and lower precision than solving the entire problem in one go. Despite these factors, our results for the precision, Gini index and entropy measurements follow the trends found in the medium sized datasets. 
\else
\textbf{Large Dataset.}
Although we omit it here due to space constraints, the large dataset we tested deserves a special discussion. In particular, the MovieLens-10m dataset has 7 times as many users as it has items in its catalog, compared to only a 1.5 times ratio in the MovieLens-1m dataset. With 70 times as many recommendations as items, even the standard recommendation approach achieved 80\% coverage of the item catalog. This renders aggregate diversity maximization ineffective as a method for increasing the equitability of the recommendation distribution. However, this fact does not effect our discrepancy method or the rerankers and the results we found are in line with the results in Table \ref{table:MLMetrics}.
\fi
\subsection{Effect of Recommendation List Length}

With more computer usage shifting from devices with larger displays like desktops and laptops to mobile devices like phones and tablets, display constraints that govern the number of recommendations we can make to user have gotten tighter. Therefore, it has become increasingly important for diversification approaches to be effective even when there is space for only a few recommendations to be made. In the set of graphs below, we fix the underlying subgraph to be the graph generated by our \ifdefined\THESIS Matrix Factorization \else MF \fi recommender with 200 candidate recommendations for each user, and measure the performance of the different diversifiers with display constraints set to $c=5, 10$ and $20$.

\begin{figure}
\centering
\begin{subfigure}[b]{.3\linewidth}
\includegraphics[width=\linewidth]{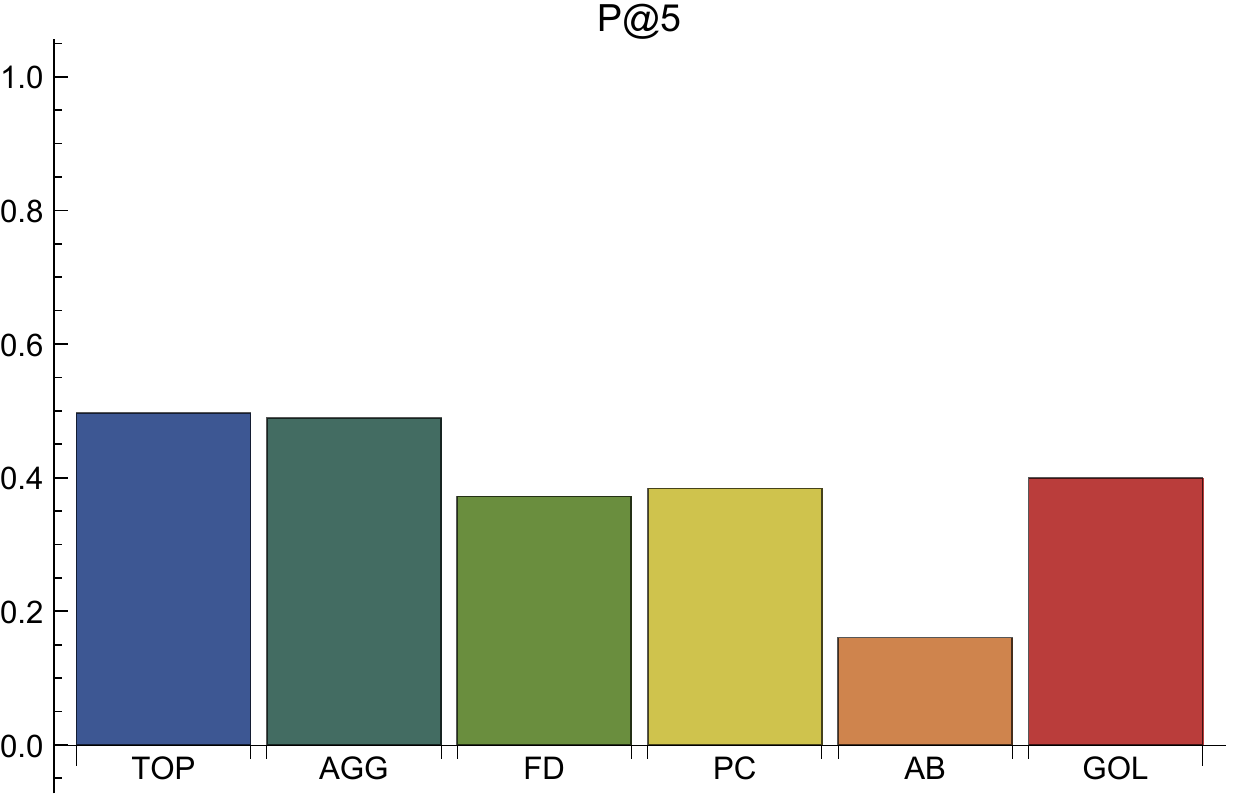}
\end{subfigure}
\begin{subfigure}[b]{.3\linewidth}
\includegraphics[width=\linewidth]{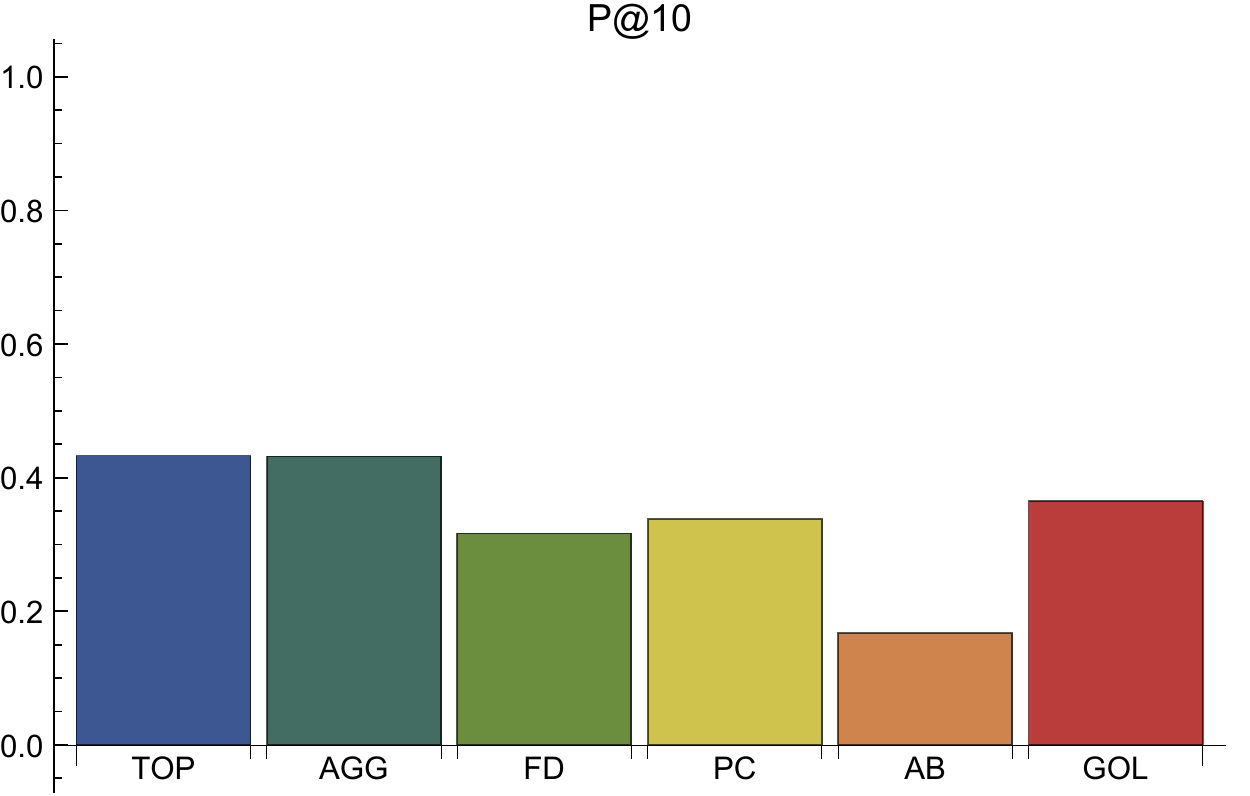}
\end{subfigure}
\begin{subfigure}[b]{.3\linewidth}
\includegraphics[width=\linewidth]{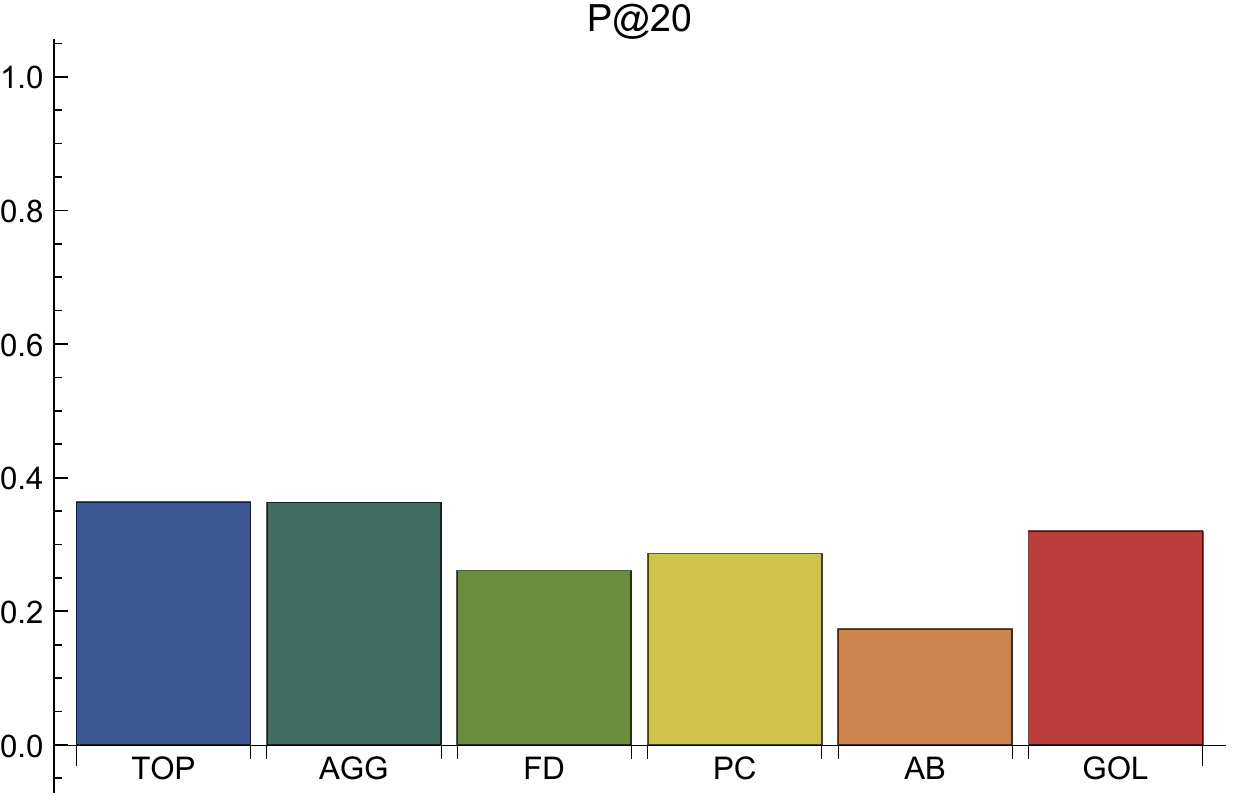}
\end{subfigure}
\vspace{.5cm}

\begin{subfigure}[b]{.3\linewidth}
\includegraphics[width=\linewidth]{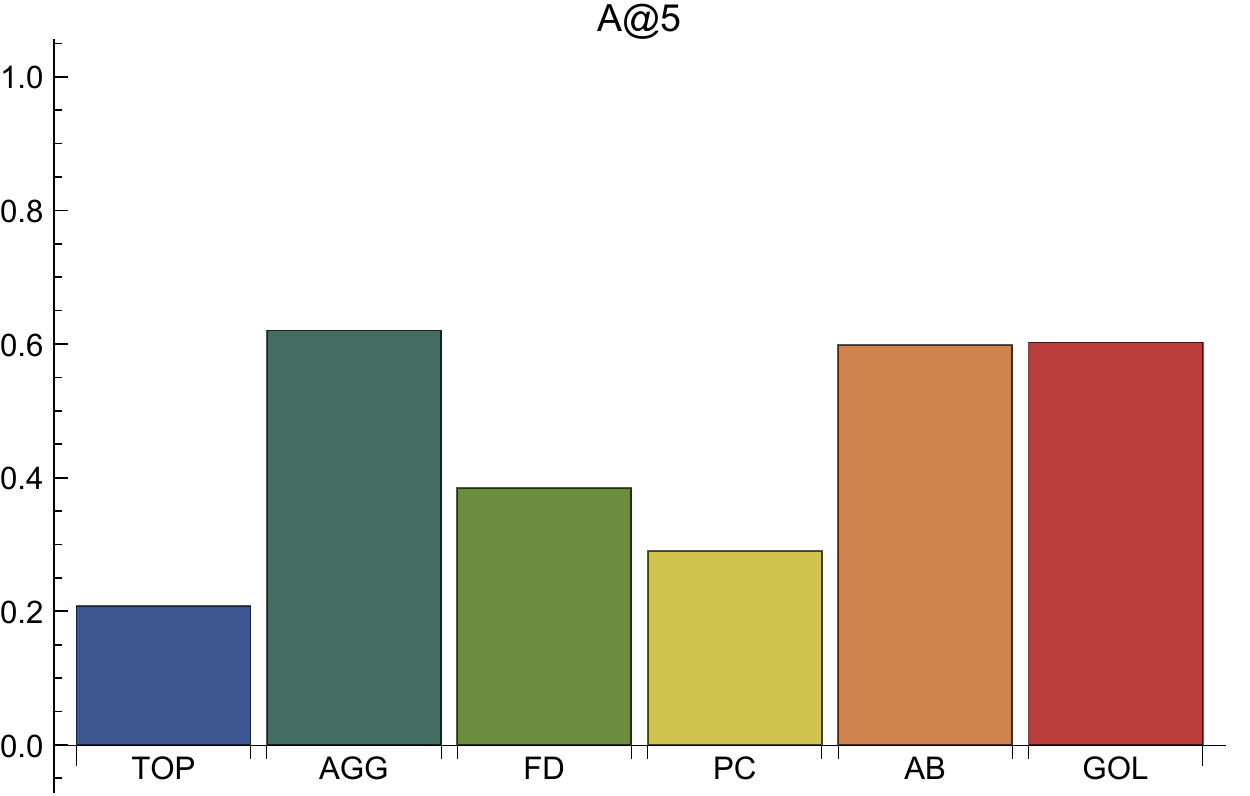}
\end{subfigure}
\begin{subfigure}[b]{.3\linewidth}
\includegraphics[width=\linewidth]{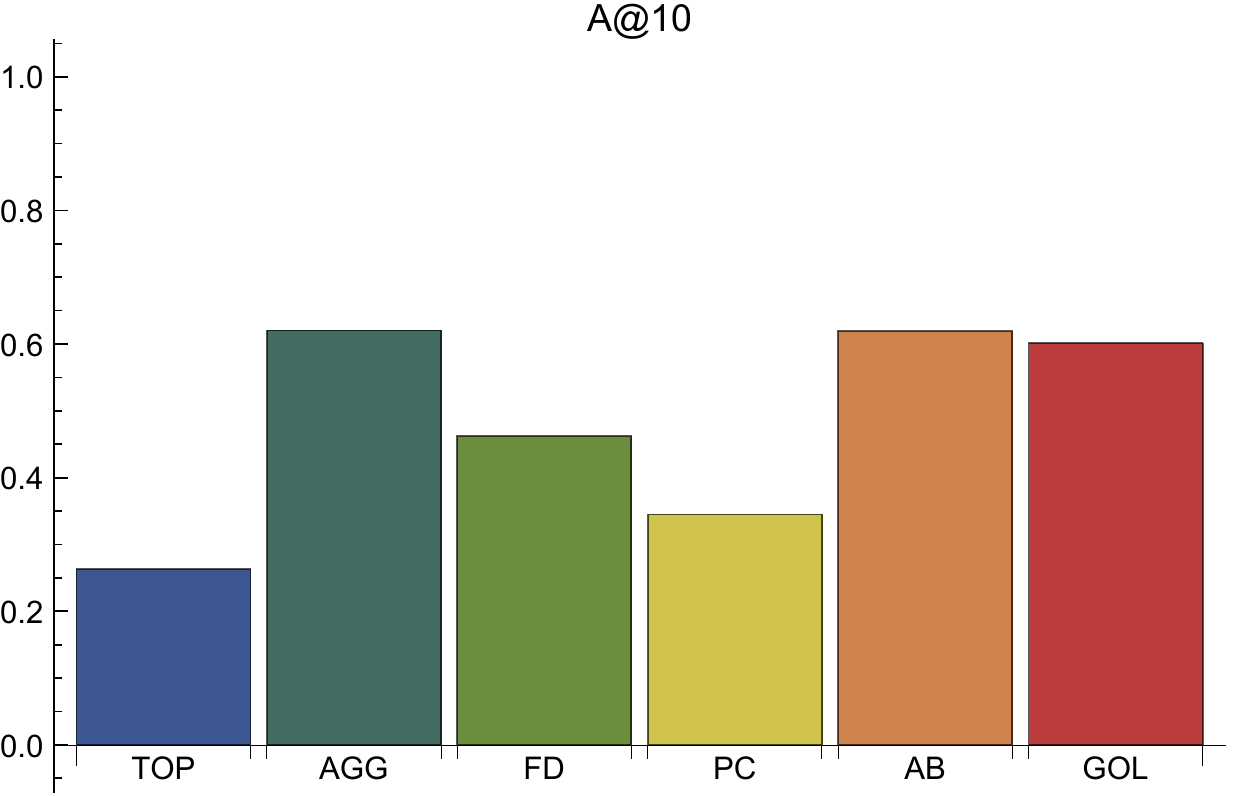}
\end{subfigure}
\begin{subfigure}[b]{.3\linewidth}
\includegraphics[width=\linewidth]{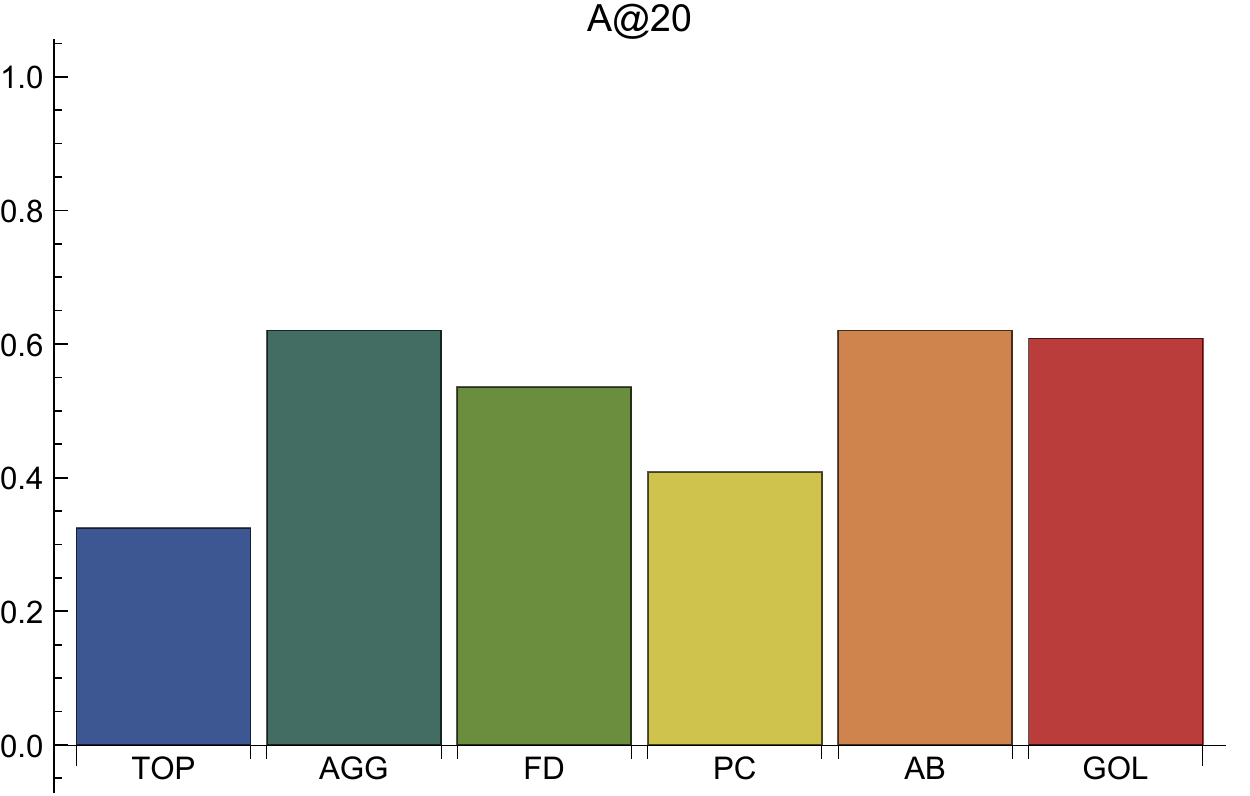}
\end{subfigure}
\vspace{.5cm}

\begin{subfigure}[b]{.3\linewidth}
\includegraphics[width=\linewidth]{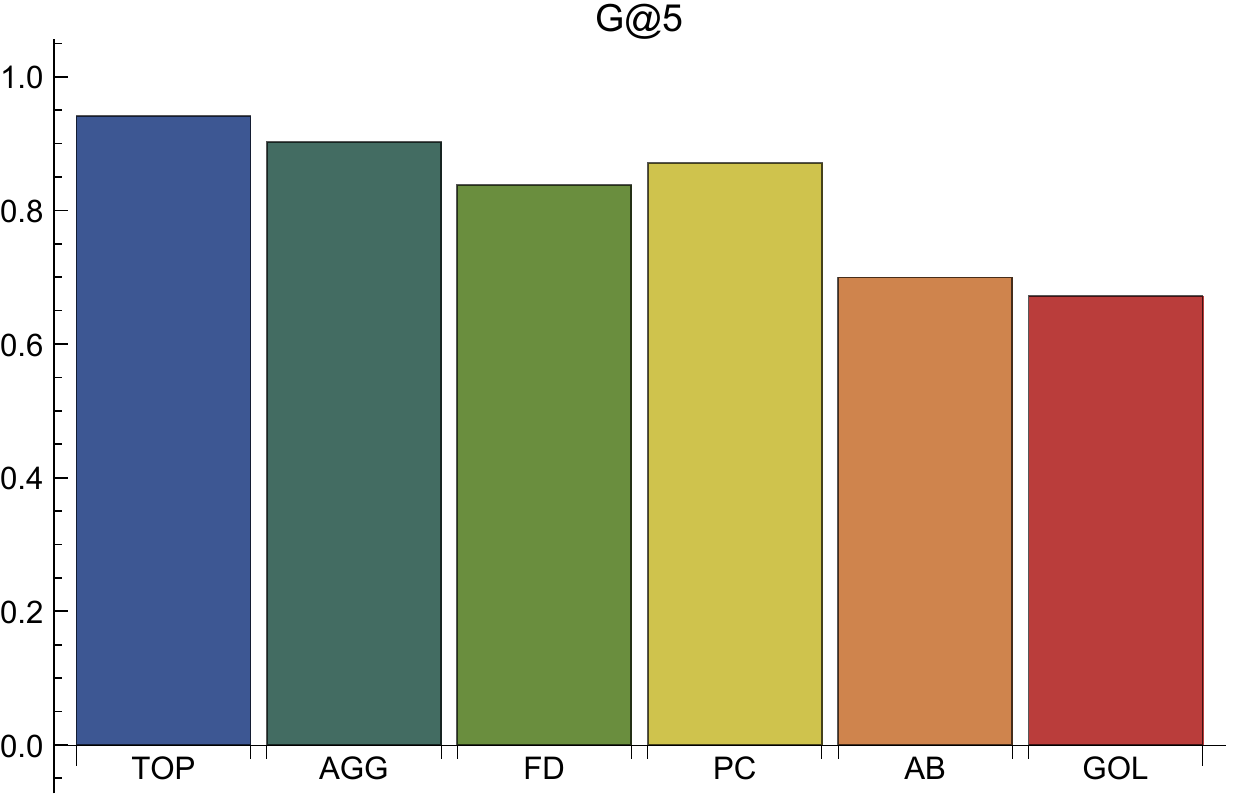}
\end{subfigure}
\begin{subfigure}[b]{.3\linewidth}
\includegraphics[width=\linewidth]{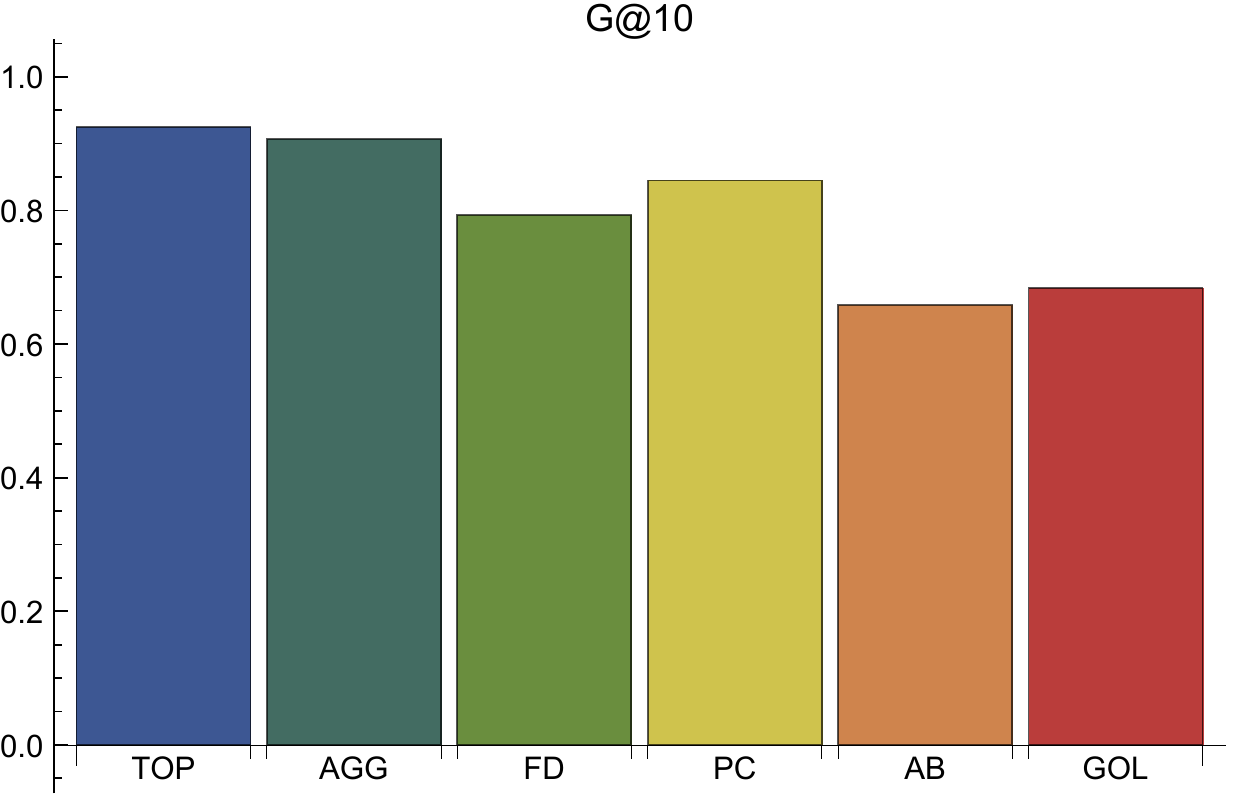}
\end{subfigure}
\begin{subfigure}[b]{.3\linewidth}
\includegraphics[width=\linewidth]{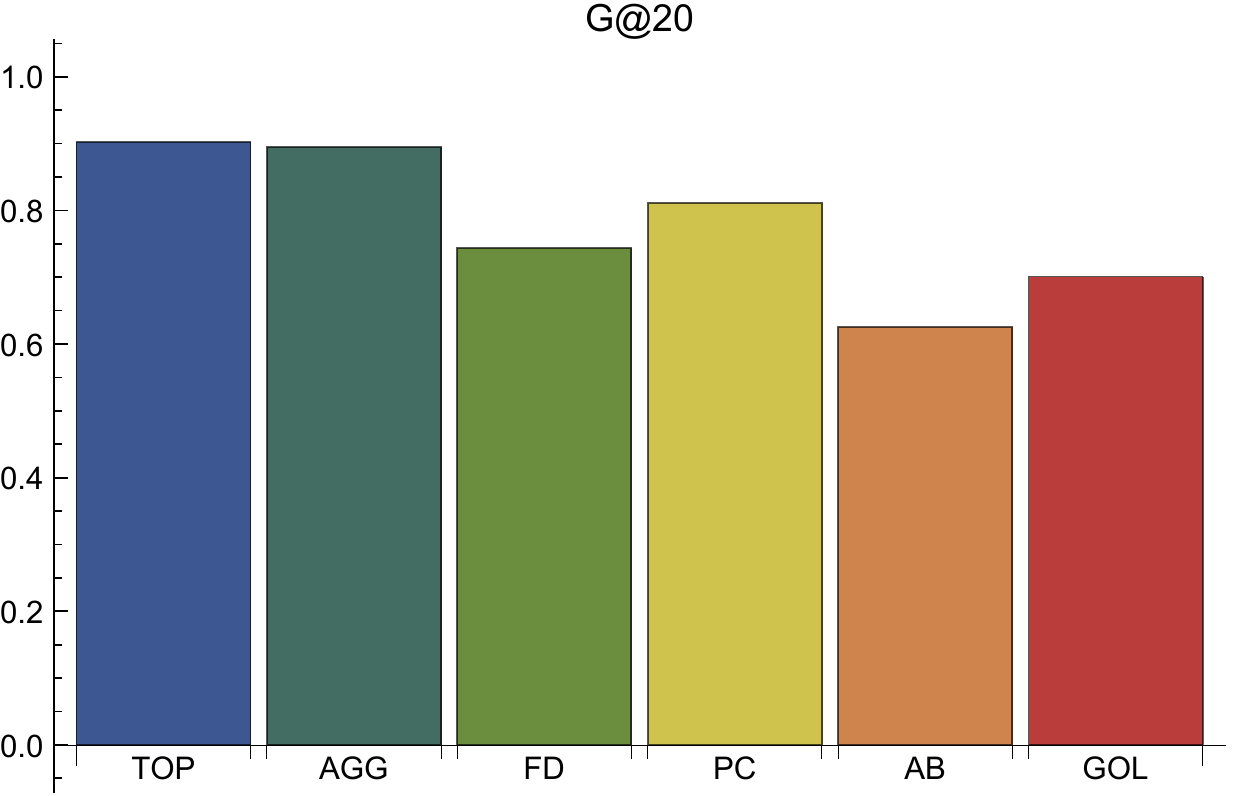}
\end{subfigure}
\vspace{.5cm}

\begin{subfigure}[b]{.3\linewidth}
\includegraphics[width=\linewidth]{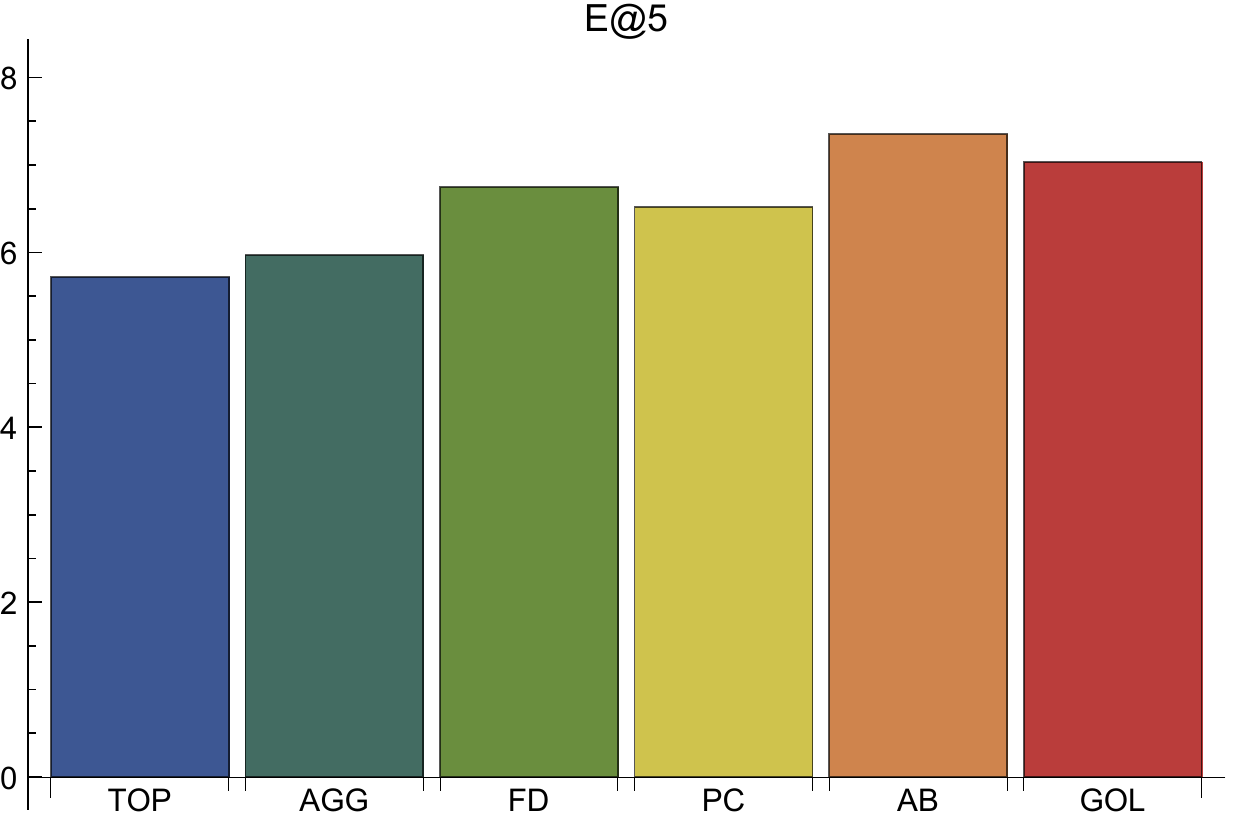}
\end{subfigure}
\begin{subfigure}[b]{.3\linewidth}
\includegraphics[width=\linewidth]{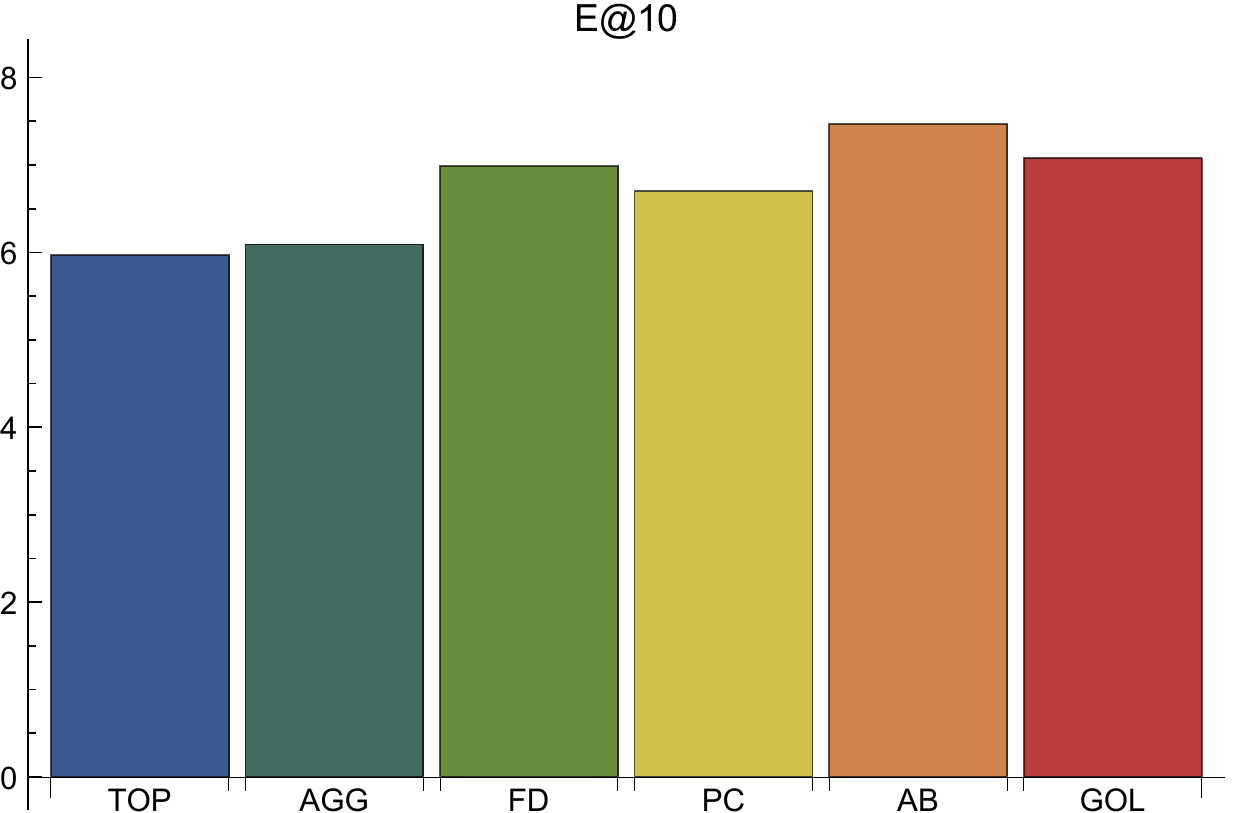}
\end{subfigure}
\begin{subfigure}[b]{.3\linewidth}
\includegraphics[width=\linewidth]{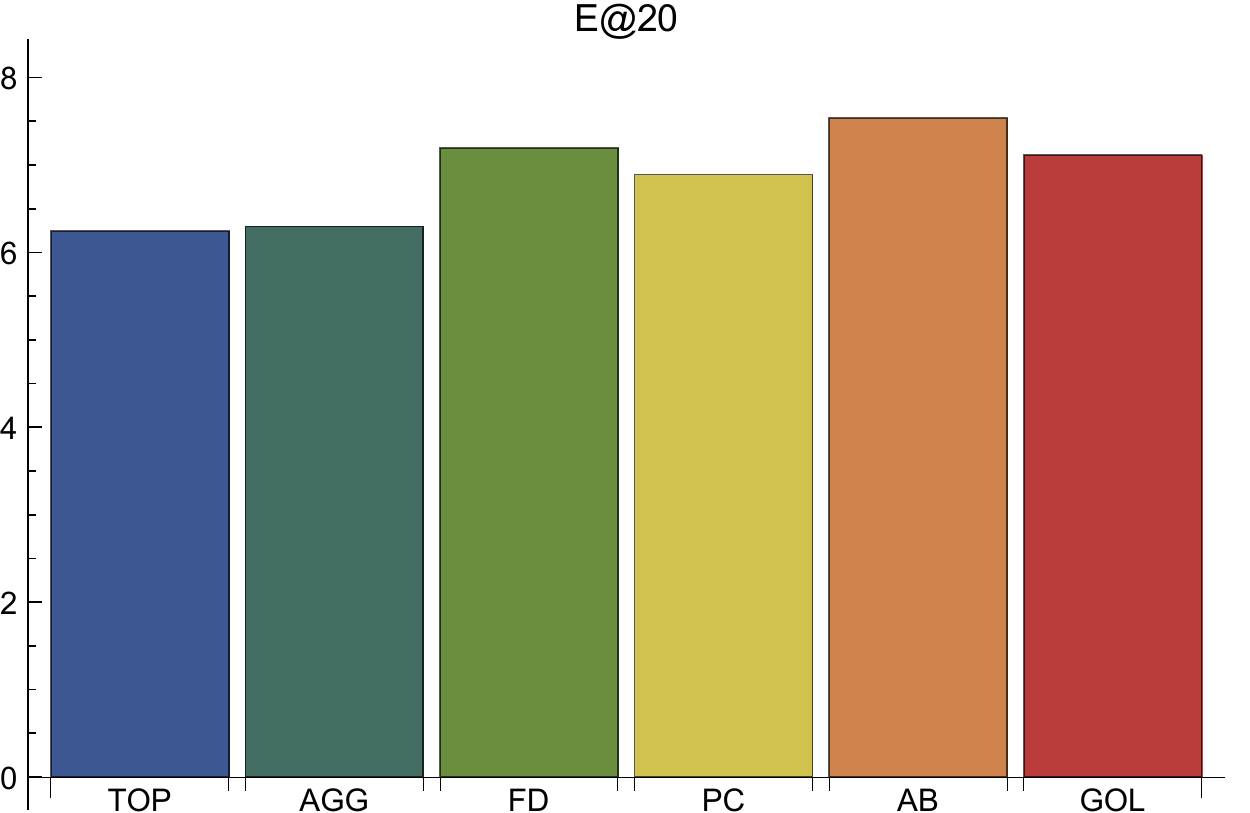}
\end{subfigure}
\caption{The effect of recommendation list length on distributional diversity measures in the MovieLens Matrix Factorization graph with 200 candidate recommendations.}
\label{fig:listlength}
\vspace{-.5cm}
\end{figure}

We note that our diversifier performs better as the display constraints get tighter. All approaches suffered precision loss as display constraints were tightened, which is to be expected. The lowest level of losses came from the Bayes Reranking approach, but this approach had such low baseline precision that all the other algorithms kept outperforming it even as recommendation lists got longer. Moreover, our two-pass method increased both its absolute and relative edge over the \FD and \PC diversifiers as $c$ was set lower.

Our two-pass method also performed better in increasing diversity metrics when the (initial) recommendation lists were shorter. In particular, there was no significant changes in the aggregate diversity achieved by our algorithms as $c$ was varied. This was also the case for Bayes Rule Reranking approach, but not for the novelty based reranking approaches FD and PC. These approaches would have needed lists of size $c=50$ to achieve comparable aggregate diversity values, at which point the precision of these diversifiers suffer considerably.

For the top-$5$ recommendation task, our method effectively outperformed all other diversifiers. \ifdefined\THESIS Our two-pass method loses out to the aggregate diversity and the baseline approaches in precision, which achieve low diversity values. It loses out to aggregate diversity maximization when considering the A@5 metric, but only by a small and statistically insignificant margin. We also achieve the highest G@5 value, while beating the runner-up Bayes Rule Reranker significantly in precision. We only achieve the second highest E@5 value, but once again beat our main competition in the Bayes Rule reanker in precision. \fi Considering all of these factors, we conclude that an optimization based framework works better in applications where display constraints are particularly tight. The reranking approaches make recommendations for each user independently, whereas our optimization framework makes all of these recommendations at once, while optimizing explicitly for a diversity metric. Therefore, our two-pass method is better able to coordinate a small number of recommendations to make large gains in diversity while also keeping precision high. 
\subsection{Trading Off Discrepancy and Precision}
\label{subsection:wgt}
In this section we explore the discrepancy/precision trade-offs made by our different models. Throughout the section, we consider the discrepancy from the uniform target distribution. This target indegree distribution sets $a_j=c\cdot l/r$ for each $v_j \in R$ for a fixed $c$ which represents the display constraints. The discrepancy from this target can be thought of as an extreme measure of diversity, since we are measuring distance from the most equitable distribution.

For the set of comparison graphs in Figure~\ref{fig:comp}, we increase the number of candidate recommendations for each user from 100 to 500 in increments of 100, and show how this affects the recommendation quality of our solution as well as the lowest discrepancy achievable. We plot the normalized discrepancy to the uniform target on the x-axis against precision in the y-axis. Discrepancy improves towards the left, and recommendation quality improves towards the top.



We consider 4 different approaches to reducing discrepancy. The first is our two-pass method, which optimizes for discrepancy first, then for average predicted relevance across the system. We also consider the weighted method with the settings $\mu = 1$, and $\mu=1/2$ and $\mu=0.01$. Recall that the weighted method optimizes the objective $\text{discrepancy}(H) - \mu\cdot\text{rel}(H)$, where $\text{rel}(H)$ denotes the average recommendation quality in the solution graph $H$ we produce. Therefore, the weighted method does not optimize for discrepancy. Instead, it find a solution where the cost of reducing discrepancy by $\mu$ units is the same as reducing aggregate predicted relevance by 1 unit.  When run on the same input graph, the predicted relevance of the two-pass method is always lower than the predicted relevance of the $\mu=1$ model, and the predicted relevance of the weighted method with $\mu$ is always lower than the predicted relevance of the weighted method with $\mu'>\mu$. The discrepancies produced by these algorithms on the same graph are also ordered in the same way.
 


\begin{figure}
\centering

\ifdefined\THESIS
\begin{subfigure}[b]{.45\linewidth}
\includegraphics[width=\linewidth]{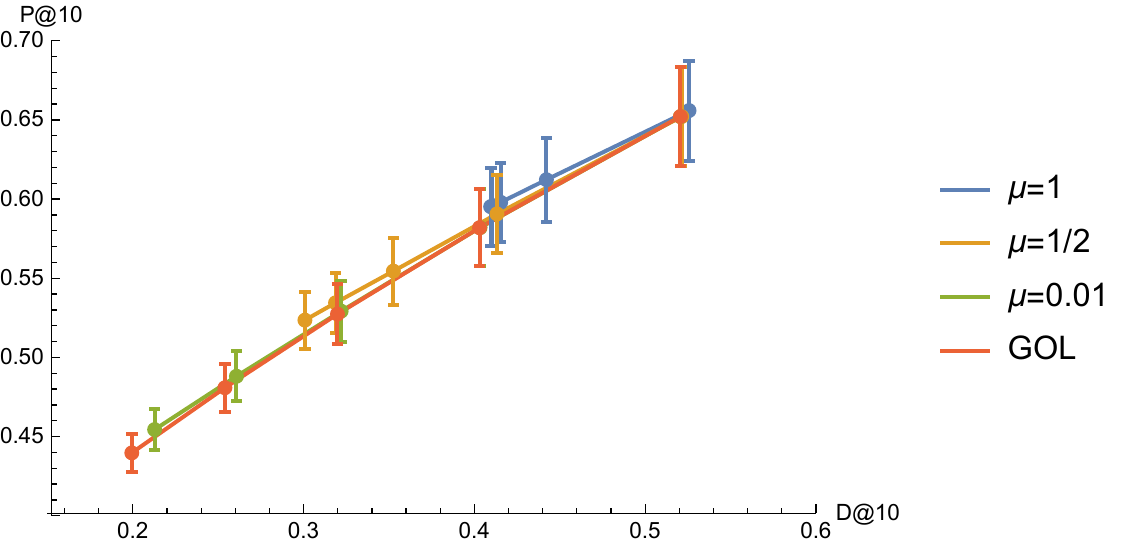}
\caption{Netflix-MF}\label{fig:mouse}
\end{subfigure}
\begin{subfigure}[b]{.45\linewidth}
\includegraphics[width=\linewidth]{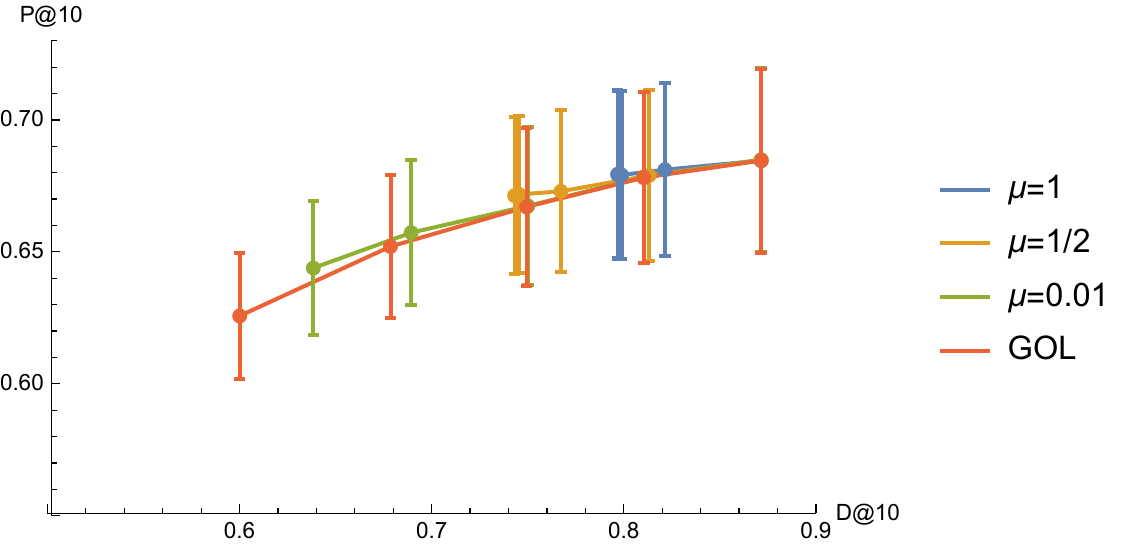}
\caption{Netflix-IB}\label{fig:mouse}
\end{subfigure}

\begin{subfigure}[b]{.45\linewidth}
\includegraphics[width=\linewidth]{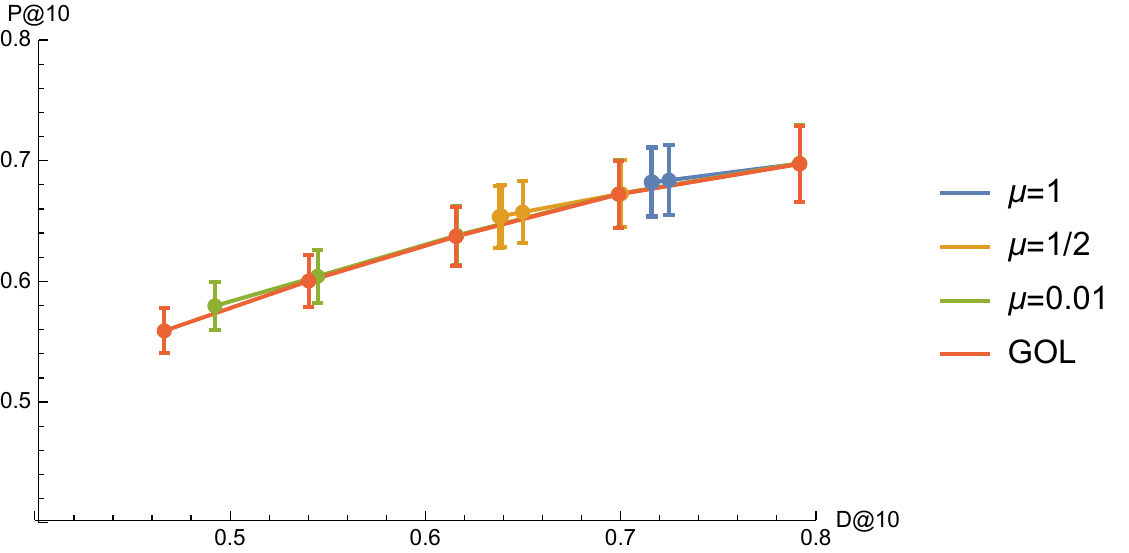}
\caption{Netflix-UB}
\end{subfigure}
\begin{subfigure}[b]{.45\linewidth}
\includegraphics[width=\linewidth]{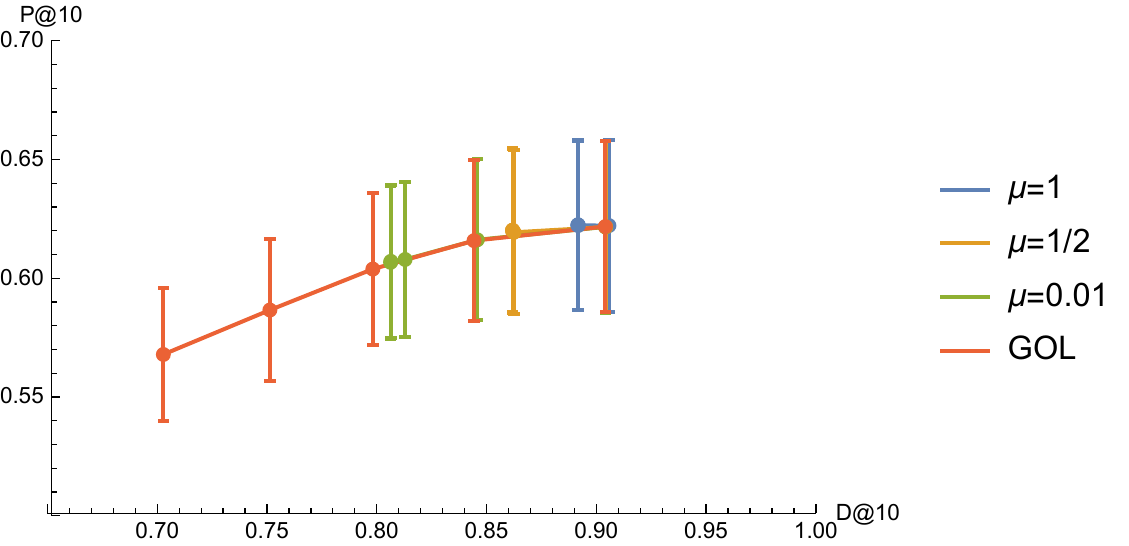}
\caption{Netflix-RW}
\end{subfigure}
\fi

\begin{subfigure}[b]{.45\linewidth}
\includegraphics[width=\linewidth]{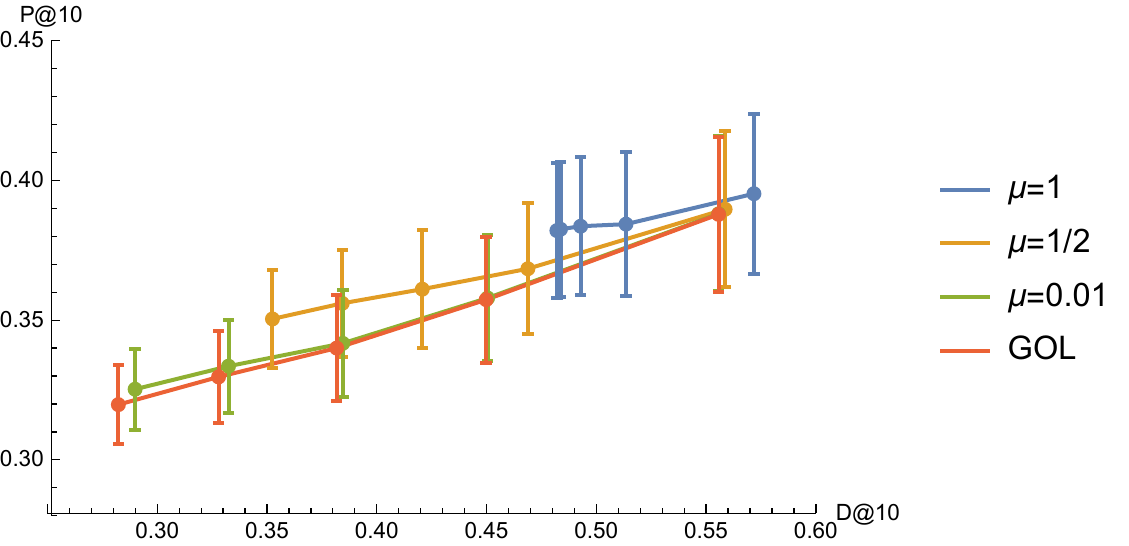}
\caption{MovieLens-MF}
\end{subfigure}
\begin{subfigure}[b]{.45\linewidth}
\includegraphics[width=\linewidth]{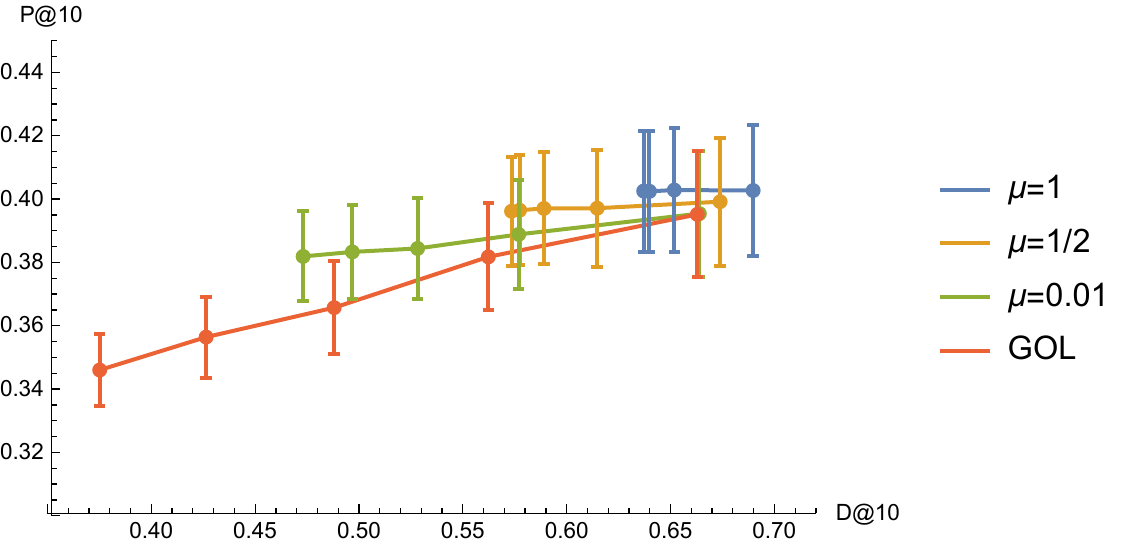}
\caption{MovieLens-IB}
\end{subfigure}

\begin{subfigure}[b]{.45\linewidth}
\includegraphics[width=\linewidth]{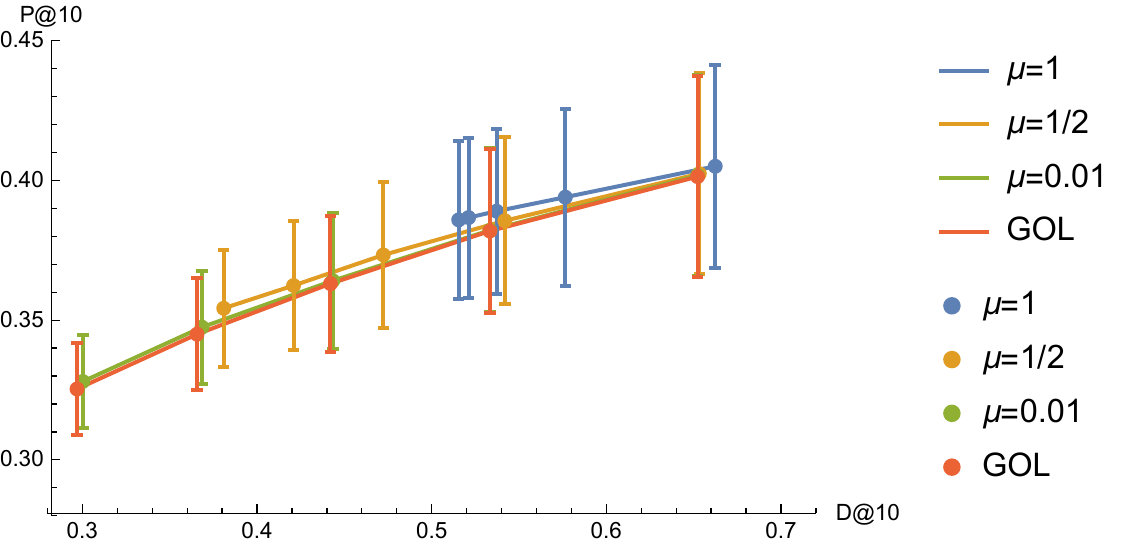}
\caption{MovieLens-UB}
\end{subfigure}
\begin{subfigure}[b]{.45\linewidth}
\includegraphics[width=\linewidth]{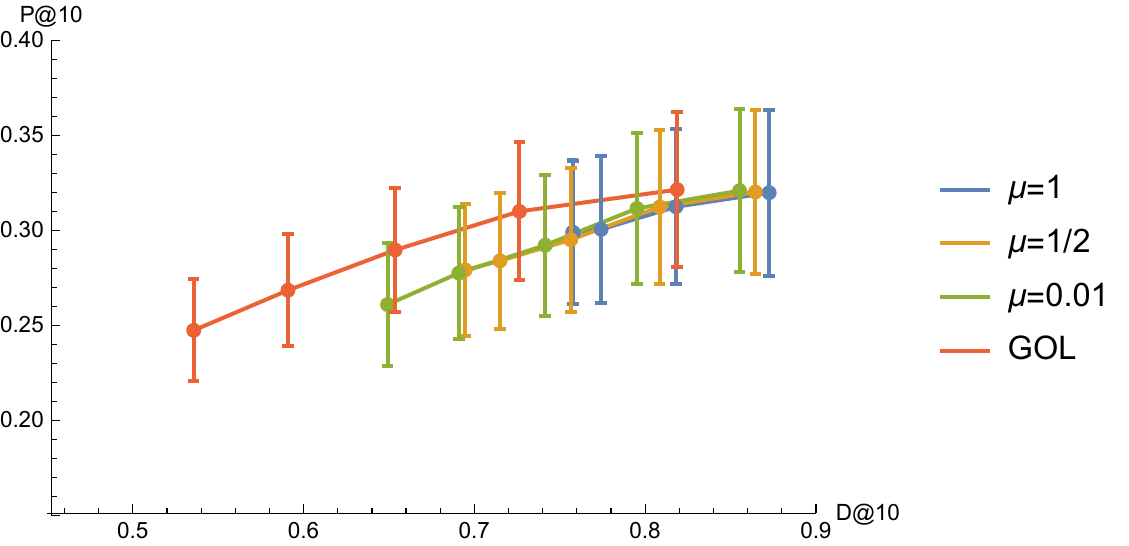}
\caption{MovieLens-RW}
\end{subfigure}

\caption{Recommendation quality vs normalized discrepancy from the uniform target in MovieLens and Netflix generated graphs. In each series, the number of edges in the input graph increases towards the left.}
\ifdefined\THESIS
\else
\vspace{-.5cm}
\fi
\label{fig:comp}
\end{figure}

From these graphs we can immediately notice certain features. First, all three algorithms produce highly clustered data with initial normalized discrepancy from the uniform  distribution always being over $0.8$. Second, the fall-off in discrepancy happens first quickly, then slowly as more edges are included. Therefore, significant gains are possible even while including a small number of candidate recommendations. Third, the choice of recommender matters. It is harder to improve the discrepancy on graphs generated by the Random Walk recommender than on any of the other other graphs we generated.

\ifdefined\THESIS
Fourth, with a high enough number of edges, discrepancy to the uniform target can be driven down towards 0. However, this may come at significant quality loss in recommendation quality depending on the underlying graph. In particular, after the addition of 500 candidate recommendations for each user, we were able to reduce discrepancy from the uniform distribution on every one of the graphs except for the Random Walk Recommender by nearly 50\% over the baseline. This drops the number of relevant recommendations surfaced by our method by 15\%-30\%. These precision losses are smaller than the losses produced the reranking diversifiers we tested, even when they were supplied with a shorter candidate recommendation lists which contain more relevant recommendations (see Table~\ref{table:MLmetrics}). Nonetheless, this level of loss in precision is not insignificant, and we suggest that users of our methods explore the full range of the trade-offs possible before putting the algorithms into commercial use.

\noindent\textbf{Two-pass versus Weighted Methods.}
\fi
One notable difference between the weighted methods and the two-pass method is that in the weighted method, the discrepancy improvements start slowing down as more and more candidate edges are included in the supergraph. As mentioned above, the bicriteria objectives improve when discrepancy gains can be made which offset the fall in predicted relevance. As we enlarge the candidate set of recommendations, we enable discrepancy increases with edges that are less and less able to make up for the corresponding relevance losses. Therefore, the solution graph stops changing though lower discrepancies are possible. Where this limiting point lies depends on the structure of the graph and the distribution of relevance values assigned by the underlying recommender and is not easy to predict. We find that the Matrix Factorization and User Based recommenders were more amenable to bicriteria optimization than the Item Based or Random Walk Based recommenders. We also find that for suitably low values of $\mu$, the weighted method can adequately approximate the output of the two-pass method, achieving essentially the same trade-off curves.
\subsection{Qualitative Parameter Tuning: Choice of Target Distribution}
\label{subsection:target}
It might not be realistic to set the same target indegree for every item since the popular items are probably popular for a good reason. In order to create alternative distribution, we will use a convex combination of two different distributions. The first distribution we use is $f$, the uniform target indegree distribution. We also generate a more skewed distribution $p$ by taking the indegree distribution of the candidate supergraph, and normalizing it so that it sums to the same value as the flat distribution. This distribution, while more diversified than the top $c$ distribution, should still be significantly unbalanced. Finally, we generate the distribution $d_\alpha = \alpha f + (1-\alpha)p$ to be used as our target. Since both $f$ and $p$ satisfy our feasibility criterion $\sum a_i = \sum c_j$, so does $d_\alpha$ for all $\alpha\in [0,1]$. Furthermore, for $0\leq\alpha \leq 1$, this function smoothly interpolates between $f$ and $p$ and produces a non-negative vector of target indegree values.

\ifdefined\THESIS
The particular setting of $\alpha$ is not special, and provided there are enough edges in the underlying graph, the two-pass method can achieve significant improvements regardless of the distribution used. To demonstrate this, fix the candidate supergraph to be the graph generated by applying the Item Based recommender to the MovieLens-1m dataset and thresholding candidate recommendations to the top 300 recommendations.
Table~\ref{fig:ConvexMatrix} below shows that as we vary $\alpha$, we get significantly different target distributions in the top-10 recommendation problem.
\begin{table}
\centering
\begin{tabular}{c|c|c|c|c|c|c|}
\cline{2-7}
                                & $d_{0}$ & $d_{0.2}$ & $d_{0.4}$ & $d_{0.6}$ & $d_{0.8}$ & $d_{1}$ \\ \hline
\multicolumn{1}{|l|}{$d_{0}$}   & 0    & 0.17 & 0.38 & 0.59 & 0.81 & 0.86                           \\ \hline
\multicolumn{1}{|l|}{$d_{0.2}$} & 0.17 & 0    & 0.21 & 0.42 & 0.64 & 0.85                           \\ \hline
\multicolumn{1}{|l|}{$d_{0.4}$} & 0.38 & 0.21 & 0    & 0.22 & 0.43 & 0.64                           \\ \hline
\multicolumn{1}{|l|}{$d_{0.6}$} & 0.59 & 0.42 & 0.22 & 0    & 0.22 & 0.42                           \\ \hline
\multicolumn{1}{|l|}{$d_{0.8}$} & 0.81 & 0.64 & 0.43 & 0.22 & 0    & 0.21                           \\ \hline
\multicolumn{1}{|l|}{$d_{1}$}   & 0.86 & 0.85 & 0.64 & 0.42 & 0.21 & 0                           \\ \hline
\end{tabular}
\caption{Pairwise discrepancy between different target distributions in the top-10 recommendation task in the MovieLens-1mItem-Based recommender and thresholded to 300 candidate recommendations}
\label{fig:ConvexMatrix}
\end{table}

Table~\ref{ConvexImprovement} shows that our diversifier achieves significant normalized discrepancy reduction in each case with the same modest rating loss.

\begin{table}
\centering
\begin{tabular}{c|c|c|c|c|c|c|}
\cline{2-7}
                                          & $d_{0}$ & $d_{0.2}$ & $d_{0.4}$ & $d_{0.6}$ & $d_{0.8}$ & $d_{1}$ \\ \hline
\multicolumn{1}{|l|}{$\Delta\text{P@10}$}    & 0.097 & 0.100 & 0.105 & 0.109 & 0.105 & 0.103                     \\ \hline
\multicolumn{1}{|l|}{$\Delta\text{D@10}$}   & 0.300 & 0.351 & 0.413 & 0.464 & 0.475 & 0.480                     \\ \hline
\end{tabular}
\caption{Rating loss vs reduction in discrepancy given different target distributions when compared with the top 10 recommendations, in the MovieLens-1m Item-Based recommender and thresholded to 200 candidate recommendations}
\label{ConvexImprovement}
\end{table}
\fi

\begin{figure}
\centering
\includegraphics[width=.99\columnwidth]{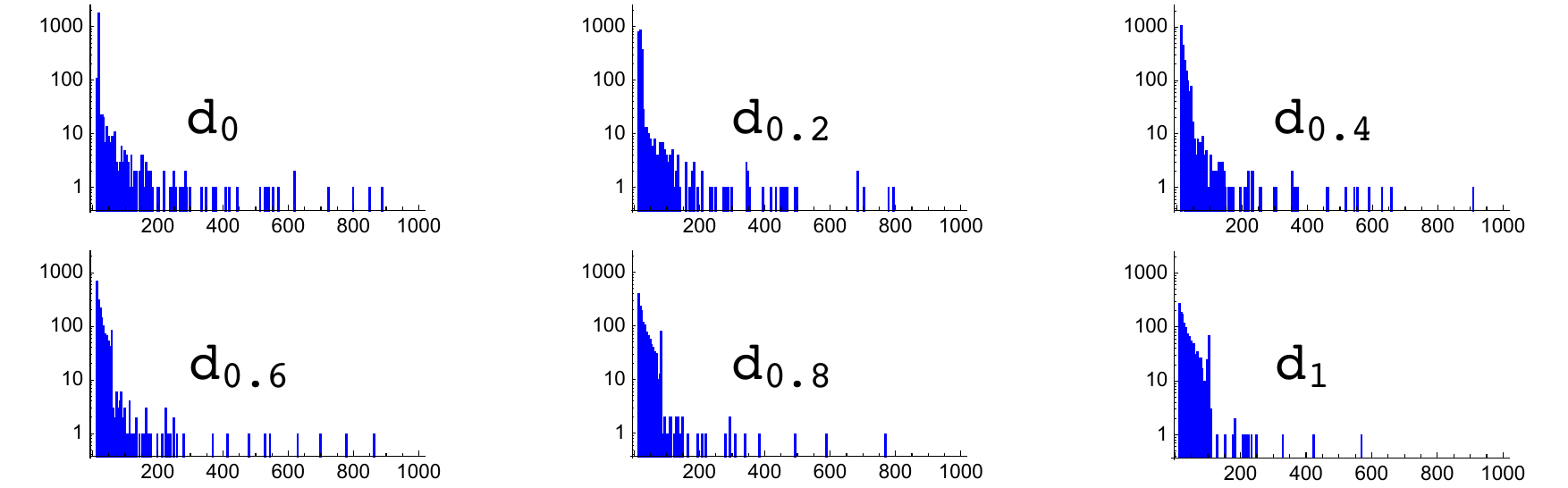}
\caption{Degree distribution in a log-scale of the solution subgraphs as the $\alpha$ of the target distribution is varied in a top-10 recommendation task. \ifdefined\THESIS The underlying supergraph is the MovieLens graph generated by Item Based recommender and thresholded to the top 200 recommendations. \fi Note the presence of large outliers when the target distribution is close to uniform.}
\label{fig:ConvexDegreeCDF}
\end{figure}

\ifdefined\THESIS
We have mentioned earlier that it is not always desirable to use the uniform target. In our modeling section, we gave a semantically motivated reason for this: a recommendation engine that recommends every item the same number of times would do a poor job of endorsing any item, even the ones that obviously deserve such endorsement. Here, we give an empirically motivated reason as well. When targeting the uniform distribution ($d_0$), discrepancy can do a poor job of producing a pleasing in-degree distribution at the items. We reduce the degrees of more vertices, more drastically this way, but this can mean that we produce a small group of vertices with really high indegree as well. In the Figure ~\ref{fig:ConvexDegreeCDF}, we show the in-degree distributions of the items in the solutions produced by our two-pass method on the same underlying graph, but with varying target distributions. It demonstrates that the degree distribution is much smoother in the global sense when a proportional target is used, and that the number of outliers with very large degree are reduced.
\else
An empirically motivated reason for using these distributions is provided in Figure ~\ref{fig:ConvexDegreeCDF}, where we show the degree distributions produced by our two-pass method on the same underlying graph but with different targets. When targeting the uniform distribution ($d_0$), discrepancy can do a poor job of producing a pleasing in-degree distribution at the items. We reduce the degrees of more vertices, more drastically this way, but this can mean that we produce a small group of vertices with really high indegree as well. Using a smoothed target distribution cuts down on the number of such outliers, even though the resulting solution has higher discrepancy from the uniform target.
\fi
\ifdefined\THESIS

\subsection{Qualitative Parameter Tuning: Convex Cost Functions}
As noted in the Algorithms section, we can change the way we charge for degree overruns
and change the types of distributions that the flow model prefers. This can be used
to remedy a particular behavior of the $\ell_1$ norm, which is used in the discrepancy definition, to prefer sparse solutions.
The flipside is that the majority of the contributions to discrepancy come from relatively few vertices. Since the amount a vertex can contribute to discrepancy by undershooting its target indegree is bounded, these vertices must in fact be very high degree vertices. Therefore, while straightforward $\ell_1$ norm minimization still has benefits in spreading degree more equitably among lower degree vertices, it does not necessarily cut down on the long tail. While it is possible to use another $\ell_p$ norm in the objective by approximating the costs by piecewise linear cost functions \cite{vegh2012strongly}, common network flow solvers neither provide tools to convert arbitrary convex costs to piecewise linear costs, nor do they implement algorithms that can efficiently solve convex cost versions of flow problems.

We can overcome this problem by more strictly punishing degree overruns. To demonstrate this, we fix our attention to the supergraph generated on the Netflix data using User Based Filtering, thresholded to the top 300 recommendations. We then diversified the graph against the target indegree distributions $d_{0},d_{1/3},d_{2/3},d_1$, as defined in the previous section. In one set of runs, we optimized our flow model directly for discrepancy. In another set of runs, labelled ``2-slope" in the chart below, we added a second sink, which does not charge for the first $a_v$ edges, charges a cost for the next $20$ edges, and charges double the cost for any edges beyond that.  The name is due to the fact that the cost of a network edge increases first with slope 1, then with slope 2, as opposed to the usual setting where it increases with slope 1 throughout. On the x-axis, we put the number of times an item is recommended. On the y-axis, we show the fraction of recommendations in the system which were made to items with fewer than a given number of recommendations.

\begin{figure}
\centering
\includegraphics[width=.7\textwidth]{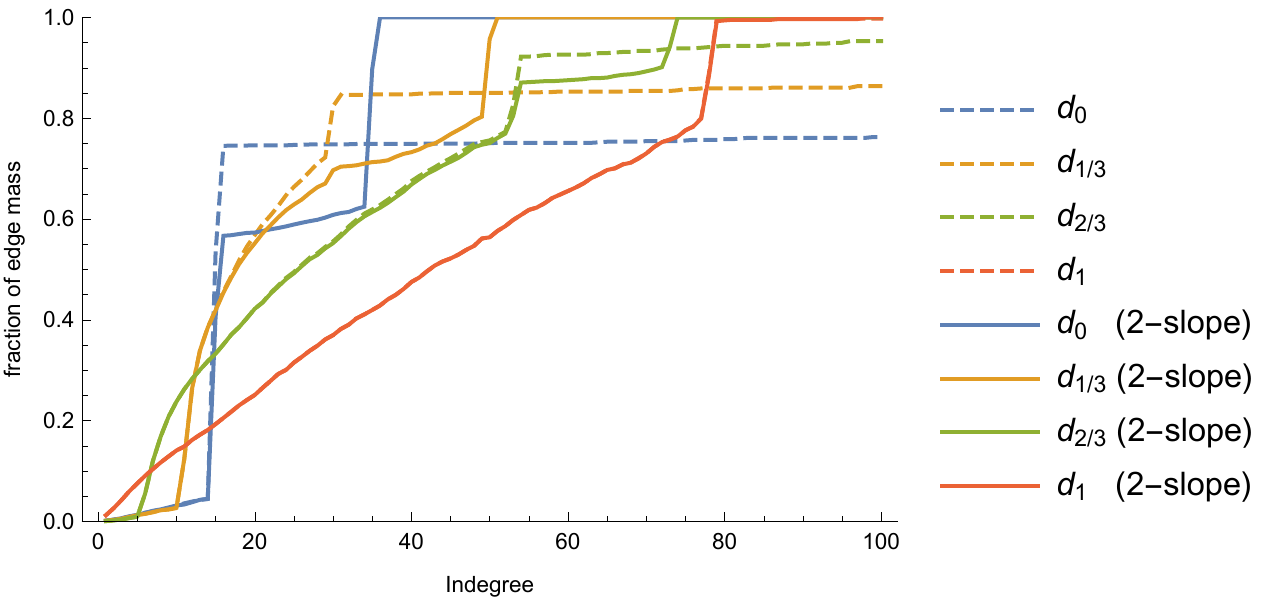}
\caption{Cumulative degree distribution of the solution subgraph with different $\alpha$ and cost functions for a candidate supergraph on the Netflix data.}
\label{fig:AggressiveLossDegreeCDF}
\end{figure}

As predicted, the 2-slope runs are much better at cutting down on the long tail. Solutions optimized for regular discrepancy tend to satisfy the indegree requirement of more vertices exactly. So when $\alpha$ is small, and the target distribution close to uniform, there tends to a large bump in the degree distribution near the average indegree of the target distribution. When the extra sink is added, it causes a secondary bump. The location of this bump can be seen to be 20 units higher than the first bump, showing the reluctance of the optimizer to go above the threshold. Secondary bumps like this cause more of the long tail to be subsumed under given thresholds. However, it should be noted that this effect is not true across all target distributions. For example, in Figure~\ref{fig:AggressiveLossDegreeCDF}, we can see that there is no difference between the 2-slope and regular runs against the $d_1$ target distribution. The reason for this is that when $\alpha$ is high, the target indegrees can all be met exactly, leading to no difference between the solutions produced by the two different schemes.
\fi
\subsection{Resource Use}
Since our methods are based on minimum cost flow, the problems that result from our reduction can be solved efficiently. The two graphs below show the cost of optimizing for uniform discrepancy with the two-pass method and the weighted method in the MovieLens-1m and Netflix graphs. We increase the number of candidate recommendations from 100 to 500, and report the average runtime across the different recommenders. The labels for this plot are identical to the labels we have been using so far, with the exception of WGT, which denotes a run of the weighted method. We did not find significant runtime differences between different settings of $\mu$ for this model, and present the results for a representative run with the setting $\mu=1$. While our methods do not run as efficiently as reranking methods, they provide much better diversification, and even instances with tens of millions of candidate recommendations can be run on a desktop.

\begin{figure}[h]
\centering

\ifdefined\THESIS
\includegraphics[width=.7\textwidth]{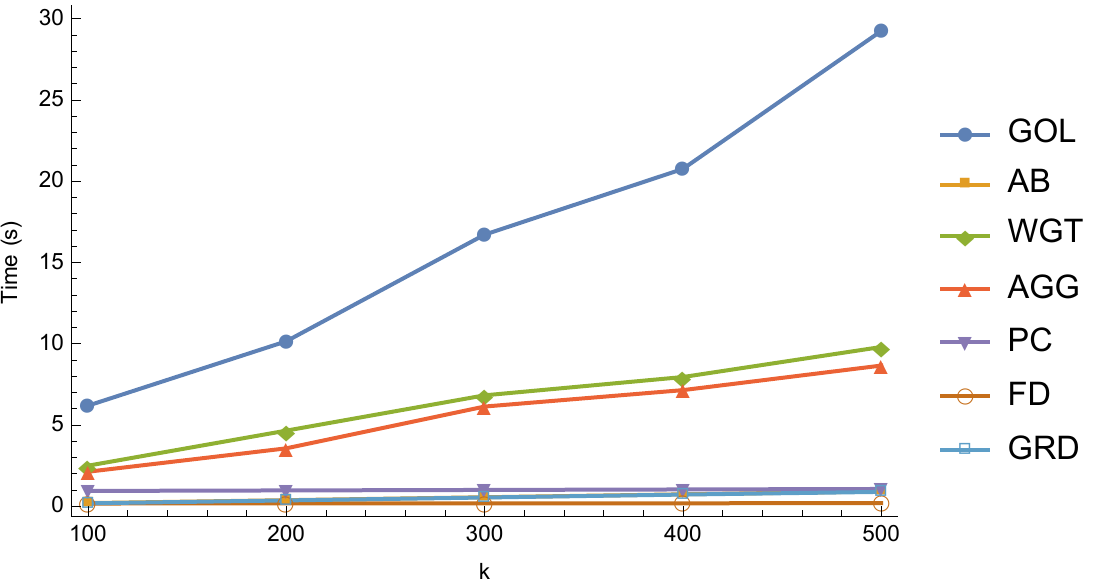}
\else
\includegraphics[width=.99\columnwidth]{chapters/chapter2/Images/MovieLensResourceComparison.pdf}
\fi

\label{fig:MLResourceUse}
\caption{Time to optimize the top-10 recommendation task in MovieLens-1m based graphs in seconds ($|L|$=5800,$|R|$=3600)}
\end{figure}

\ifdefined\THESIS
\begin{figure}[h]
\centering
\includegraphics[width=.7\textwidth]{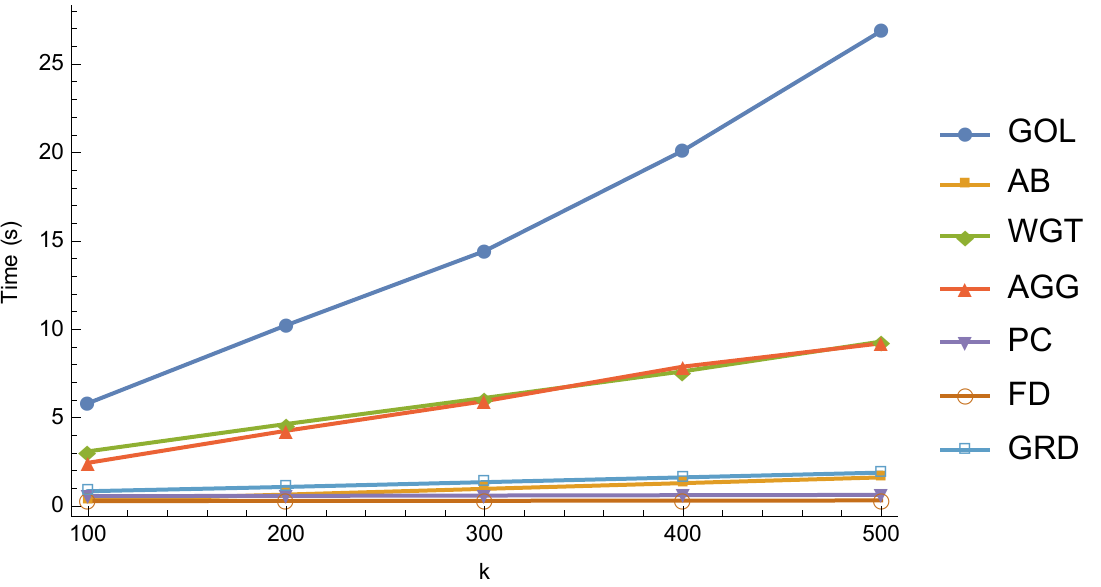}
\label{fig:NFResourceUse}
\caption{Time to optimize the top-10 recommendation task in Netflix based graphs in seconds ($|L|$=8000,$|R|$=5000)}
\end{figure}

\begin{figure}[h]
\centering
\includegraphics[width=.7\textwidth]{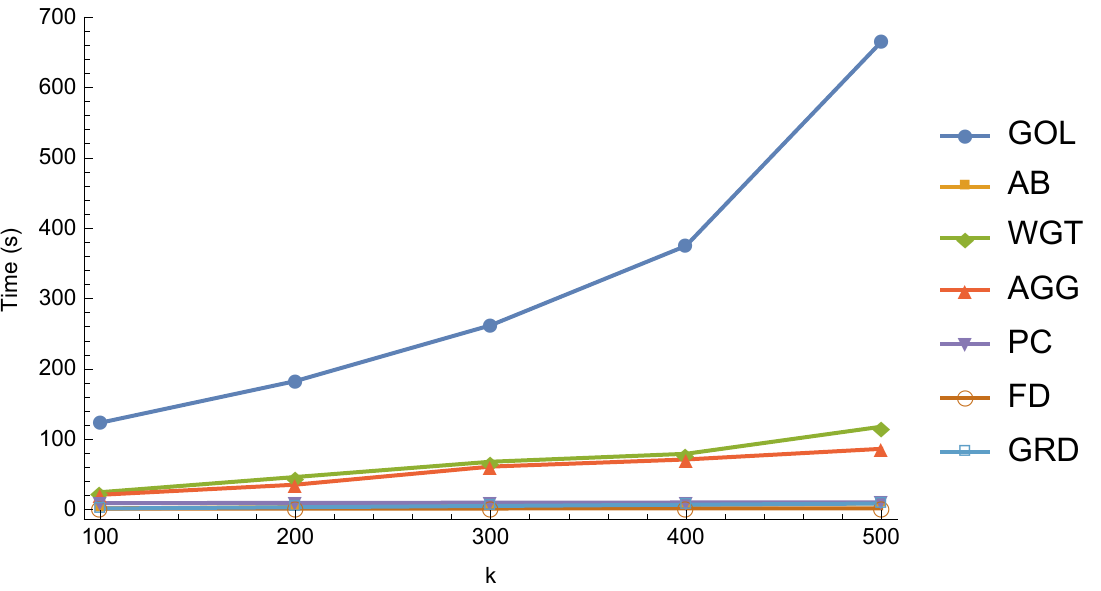}
\label{fig:MLBigResourceUse}
\caption{Time to optimize the top-10 recommendation task in Netflix based graphs in seconds ($|L|$=67000,$|R|$=9000)}
\end{figure}
\fi

\ifdefined\THESIS
We also note that the two-pass method takes noticeably longer than twice runtime the cost of the weighted method. It may be possible to improve this by a using a better implementation of the two-pass method, whereby the feasible solution found in the first pass is fed as a seed solution to the second pass. While we do not test this approach, we recommend its use to anyone who would like to use our framework in a commercial setting. Even if this improvement were to be made, the second pass of the two-pass method optimization is still a more challenging optimization problem than the first pass since it includes a stricter set of constraints. Moreover, the difficulty of the second pass is determined by how close the discrepancy target we set is to the lowest discrepancy achievable, which can lead to a significant increase in resource use. This problem can be dealt with in practice either by reducing the size of the discrepancy problems as mentioned in Section~\ref{subsection:largedataset}, or by using the weighted method instead of the two-pass method as described in Section~\ref{subsection:comparison}.
\fi

\section{Conclusions}
We have proposed a new way of measuring how equitably a recommender system distributes its recommendations called discrepancy, and showed that it can be optimized for in polynomial time using network flow techniques. We validated the effectiveness and the efficiency of our method by conducting extensive tests on MovieLens and Netflix datasets, and showed it to improve diversity across a variety of measures. Our work demonstrates that distributional diversity measures like discrepancy can be efficiently optimized to allow information designers to have more control over their recommender systems.  

\bibliographystyle{plain}  

\begin{thebibliography}{10}

\bibitem{MLData}
Movielens-1m data set.
\newblock \url{http://grouplens.org/datasets/movielens/1m/}.
\newblock Accessed: 03/2015.

\bibitem{adamopoulos2011unexpectedness}
Panagiotis Adamopoulos and Alexander Tuzhilin.
\newblock On unexpectedness in recommender systems: Or how to expect the
  unexpected.
\newblock In {\em DiveRS@ RecSys}, pages 11--18, 2011.

\bibitem{adomavicius2011maximizing}
Gediminas Adomavicius and YoungOk Kwon.
\newblock Maximizing aggregate recommendation diversity: A graph-theoretic
  approach.
\newblock In {\em Proc. of the 1st Int. Workshop on Novelty and Diversity in
  Recommender Systems (DiveRS 2011)}, pages 3--10, 2011.

\bibitem{adomavicius2012improving}
Gediminas Adomavicius and YoungOk Kwon.
\newblock Improving aggregate recommendation diversity using ranking-based
  techniques.
\newblock {\em IEEE Trans. on Knowl. and Data Eng.}, 24(5):896--911, 2012.

\bibitem{ahuja1993network}
R.K. Ahuja, T.L. Magnanti, and J.B. Orlin.
\newblock {\em Network flows: theory, algorithms, and applications}.
\newblock Prentice Hall, 1993.

\bibitem{ali2004tivo}
Kamal Ali and Wijnand Van~Stam.
\newblock Tivo: making show recommendations using a distributed collaborative
  filtering architecture.
\newblock In {\em Proc. of the 10th ACM SIGKDD Int. Conf. on Knowl. discovery
  and data mining}, pages 394--401. ACM, 2004.

\bibitem{Anderson2006}
C.~Anderson.
\newblock {\em The Long Tail: Why the Future of Business Is Selling Less of
  More}.
\newblock Hyperion, 2006.

\bibitem{antikacioglu2015recommendation}
Arda Antikacioglu, R~Ravi, and Srinath Sridhar.
\newblock Recommendation subgraphs for web discovery.
\newblock In {\em Proc. of the 24th Int. Conf. on World Wide Web}, pages
  77--87. Int. World Wide Web Conf.s Steering Committee, 2015.

\bibitem{bell2007modeling}
Robert Bell, Yehuda Koren, and Chris Volinsky.
\newblock Modeling relationships at multiple scales to improve accuracy of
  large recommender systems.
\newblock In {\em Proc. of the 13th ACM SIGKDD Int. Conf. on Knowl. discovery
  and data mining}, pages 95--104. ACM, 2007.

\bibitem{brynjolfsson2011goodbye}
Erik Brynjolfsson, Yu~Hu, and Duncan Simester.
\newblock Goodbye pareto principle, hello long tail: The effect of search costs
  on the concentration of product sales.
\newblock {\em Management Science}, 57(8):1373--1386, 2011.

\bibitem{brynjolfsson2003consumer}
Erik Brynjolfsson, Yu~Hu, and Michael~D Smith.
\newblock Consumer surplus in the digital economy: Estimating the value of
  increased product variety at online booksellers.
\newblock {\em Management Science}, 49(11):1580--1596, 2003.

\bibitem{castells2011novelty}
Pablo Castells, Sa{\'u}l Vargas, and Jun Wang.
\newblock Novelty and diversity metrics for recommender systems: choice,
  discovery and relevance.
\newblock 2011.

\bibitem{celma2008hits}
{\`O}scar Celma and Pedro Cano.
\newblock From hits to niches?: or how popular artists can bias music
  recommendation and discovery.
\newblock In {\em Proc. of the 2nd KDD Workshop on Large-Scale Recommender
  Systems}, page~5. ACM, 2008.

\bibitem{chen2016conflict}
Cheng Chen, Lan Zheng, Venkatesh Srinivasan, Alex Thomo, Kui Wu, and Anthony
  Sukow.
\newblock Conflict-aware weighted bipartite b-matching and its application to
  e-commerce.
\newblock {\em IEEE Trans. on Knowl. and Data Eng.}, 28(6):1475--1488, 2016.

\bibitem{christoffel2015blockbusters}
Fabian Christoffel, Bibek Paudel, Chris Newell, and Abraham Bernstein.
\newblock Blockbusters and wallflowers: Accurate, diverse, and scalable
  recommendations with random walks.
\newblock In {\em Proc. of the 9th ACM Conf. on Rec. Systems}, pages 163--170.
  ACM, 2015.

\bibitem{cooper2014random}
Colin Cooper, Sang~Hyuk Lee, Tomasz Radzik, and Yiannis Siantos.
\newblock Random walks in recommender systems: Exact computation and
  simulations.
\newblock In {\em Proceedings of the 23rd International Conference on World
  Wide Web}, pages 811--816. ACM, 2014.

\bibitem{desrosiers2011comprehensive}
Christian Desrosiers and George Karypis.
\newblock A comprehensive survey of neighborhood-based recommendation methods.
\newblock In {\em Recommender systems handbook}, pages 107--144. Springer,
  2011.

\bibitem{fleder2009blockbuster}
Daniel Fleder and Kartik Hosanagar.
\newblock Blockbuster culture's next rise or fall: The impact of recommender
  systems on sales diversity.
\newblock {\em Management Science}, 55(5):697--712, 2009.

\bibitem{Frangioni2010MCFSimplex}
Antonio Frangioni and Luis~Perez Sanchez.
\newblock Searching the best (formulation, solver, configuration) for
  structured problems.
\newblock In {\em Complex Systems Design \& Management}, pages 85--98.
  Springer, 2010.

\bibitem{gao2011userrank}
Min Gao, Zhongfu Wu, and Feng Jiang.
\newblock Userrank for item-based collaborative filtering recommendation.
\newblock {\em Information Processing Letters}, 111(9):440--446, 2011.

\bibitem{ge2010beyond}
Mouzhi Ge, Carla Delgado-Battenfeld, and Dietmar Jannach.
\newblock Beyond accuracy: evaluating recommender systems by coverage and
  serendipity.
\newblock In {\em Proceedings of the fourth ACM conference on Recommender
  systems}, pages 257--260. ACM, 2010.

\bibitem{goldstein2006profiting}
Daniel~G Goldstein and Dominique~C Goldstein.
\newblock Profiting from the long tail.
\newblock {\em Harvard Business Review}, 84(6):24--28, 2006.

\bibitem{gori2007itemrank}
Marco Gori, Augusto Pucci, V~Roma, and I~Siena.
\newblock Itemrank: A random-walk based scoring algorithm for recommender
  engines.
\newblock In {\em IJCAI}, volume~7, pages 2766--2771, 2007.

\bibitem{herlocker2004evaluating}
Jonathan~L Herlocker, Joseph~A Konstan, Loren~G Terveen, and John~T Riedl.
\newblock Evaluating collaborative filtering recommender systems.
\newblock {\em ACM Transactions on Information Systems (TOIS)}, 22(1):5--53,
  2004.

\bibitem{hu2008collaborative}
Yifan Hu, Yehuda Koren, and Chris Volinsky.
\newblock Collaborative filtering for implicit feedback datasets.
\newblock In {\em Data Mining, 2008. ICDM'08. Eighth IEEE International
  Conference on}, pages 263--272. IEEE, 2008.

\bibitem{karimzadehgan2009constrained}
Maryam Karimzadehgan and ChengXiang Zhai.
\newblock Constrained multi-aspect expertise matching for committee review
  assignment.
\newblock In {\em Proc. of the 18th ACM Conf. on Inform. and Knowledge
  Management}, pages 1697--1700. ACM, 2009.

\bibitem{koren2008factorization}
Yehuda Koren.
\newblock Factorization meets the neighborhood: a multifaceted collaborative
  filtering model.
\newblock In {\em Proc. of the 14th ACM SIGKDD Int. Conf. on Knowl. discovery
  and data mining}, pages 426--434. ACM, 2008.

\bibitem{kovacs2015minimum}
P{\'e}ter Kov{\'a}cs.
\newblock Minimum-cost flow algorithms: an experimental evaluation.
\newblock {\em Optimization Methods and Software}, 30(1):94--127, 2015.

\bibitem{liu2012solving}
Jian-Guo Liu, Kerui Shi, and Qiang Guo.
\newblock Solving the accuracy-diversity dilemma via directed random walks.
\newblock {\em Physical Review E}, 85(1):016118, 2012.

\bibitem{makari2013distributed}
Faraz Makari and Rainer Gemulla.
\newblock A distributed approximation algorithm for mixed packing-covering
  linear programs.
\newblock In {\em NIPS 2013 Workshop on Big Learning}. NIPS, 2013.

\bibitem{npr-segment}
Cecilia Mazanec.
\newblock Will algorithms erode our decision-making skills?, 2016.
\newblock Accessed: 02/2016.

\bibitem{mcnee2006being}
Sean~M McNee, John Riedl, and Joseph~A Konstan.
\newblock Being accurate is not enough: how accuracy metrics have hurt
  recommender systems.
\newblock In {\em CHI'06 Human factors in computing systems}, pages 1097--1101.
  ACM, 2006.

\bibitem{mimno2007expertise}
David Mimno and Andrew McCallum.
\newblock Expertise modeling for matching papers with reviewers.
\newblock In {\em Proc. of the 13th ACM SIGKDD Int. Conf. on Knowl. discovery
  and data mining}, pages 500--509. ACM, 2007.

\bibitem{WeaponsMathDestructionBook}
Cathy O'Neil.
\newblock {\em Weapons of math destruction: How big data increases inequality
  and threatens democracy}.
\newblock Crown Publishing Group (NY), 2016.

\bibitem{orlin1997polynomial}
James~B Orlin.
\newblock A polynomial time primal network simplex algorithm for minimum cost
  flows.
\newblock {\em Mathematical Programming}, 78(2):109--129, 1997.

\bibitem{parameswaran2011recommendation}
Aditya Parameswaran, Petros Venetis, and Hector Garcia-Molina.
\newblock Recommendation systems with complex constraints: A course
  recommendation perspective.
\newblock {\em ACM Trans. on Inform. Systems (TOIS)}, 29(4):20, 2011.

\bibitem{pariser2011filter}
Eli Pariser.
\newblock {\em The filter bubble: How the new personalized web is changing what
  we read and how we think}.
\newblock Penguin, 2011.

\bibitem{ren2014avoiding}
Xiaolong Ren, Linyuan L{\"u}, Runran Liu, and Jianlin Zhang.
\newblock Avoiding congestion in recommender systems.
\newblock {\em New Journal of Physics}, 16(6):063057, 2014.

\bibitem{sandoval2015novelty}
SV~Sandoval.
\newblock Novelty and diversity evaluation and enhancement in recommender
  systems, 2015.

\bibitem{sarwar2001item}
Badrul Sarwar, George Karypis, Joseph Konstan, and John Riedl.
\newblock Item-based collaborative filtering recommendation algorithms.
\newblock In {\em Proc. of the 10th Int. Conf. on World Wide Web}, pages
  285--295. ACM, 2001.

\bibitem{schrijver2002combinatorial}
Alexander Schrijver.
\newblock {\em Combinatorial optimization: polyhedra and efficiency},
  volume~24.
\newblock Springer Science \& Business Media, 2002.

\bibitem{senecal2004influence}
Sylvain Senecal and Jacques Nantel.
\newblock The influence of online product recommendations on consumers online
  choices.
\newblock {\em Journal of Retailing}, 80(2):159--169, 2004.

\bibitem{shani2011evaluating}
Guy Shani and Asela Gunawardana.
\newblock Evaluating recommendation systems.
\newblock In {\em Recommender systems handbook}, pages 257--297. Springer,
  2011.

\bibitem{szlavik2011diversity}
Zolt{\'a}n Szl{\'a}vik, Wojtek Kowalczyk, and Martijn Schut.
\newblock Diversity measurement of recommender systems under different user
  choice models.
\newblock In {\em Fifth International AAAI Conference on Weblogs and Social
  Media}, 2011.

\bibitem{tulving1994novelty}
Endel Tulving, Hans~J Markowitsch, Shitij Kapur, Reza Habib, and Sylvain Houle.
\newblock Novelty encoding networks in the human brain: positron emission
  tomography data.
\newblock {\em NeuroReport}, 5(18):2525--2528, 1994.

\bibitem{vargas2015novelty}
Sa{\'u}l Vargas.
\newblock Novelty and diversity enhancement and evaluation in recommender
  systems.
\newblock 2015.

\bibitem{vargas2011rank}
Sa{\'u}l Vargas and Pablo Castells.
\newblock Rank and relevance in novelty and diversity metrics for recommender
  systems.
\newblock In {\em Proc. of the 5th ACM Conf. on Rec. Systems}, pages 109--116.
  ACM, 2011.

\bibitem{vargas2014improving}
Sa\'{u}l Vargas and Pablo Castells.
\newblock Improving sales diversity by recommending users to items.
\newblock In {\em Proc. of the 8th ACM Conf. on Rec. Systems}, RecSys '14,
  pages 145--152, New York, NY, USA, 2014. ACM.

\bibitem{vegh2012strongly}
L{\'a}szl{\'o}~A V{\'e}gh.
\newblock Strongly polynomial algorithm for a class of minimum-cost flow
  problems with separable convex objectives.
\newblock In {\em Proc. of the 44th annual ACM symposium on Theory of
  computing}, pages 27--40. ACM, 2012.

\bibitem{wang2006unifying}
Jun Wang, Arjen~P De~Vries, and Marcel~JT Reinders.
\newblock Unifying user-based and item-based collaborative filtering approaches
  by similarity fusion.
\newblock In {\em Proc. of the 29th Int. ACM SIGIR Conf. on Research and
  development in information retrieval}, pages 501--508. ACM, 2006.

\bibitem{zhang2008avoiding}
Mi~Zhang and Neil Hurley.
\newblock Avoiding monotony: improving the diversity of recommendation lists.
\newblock In {\em Proc. of the 2008 ACM Conf. on Rec. Systems}, pages 123--130.
  ACM, 2008.

\bibitem{zhou2010impact}
Renjie Zhou, Samamon Khemmarat, and Lixin Gao.
\newblock The impact of youtube recommendation system on video views.
\newblock In {\em Proc. of the 10th ACM SIGCOMM Conf. on Internet measurement},
  pages 404--410. ACM, 2010.

\end{thebibliography}


\end{document}  